%% file: ms.tex
\documentclass[10pt,journal,twocolumn,twoside]{IEEEtran}

\usepackage{amsfonts,amsbsy,bm,amsmath}
\usepackage[nolist]{acronym}
\usepackage{mathrsfs} 
\usepackage{graphicx,cite,amssymb,mathtools,amsthm}
\usepackage{stfloats}
\usepackage{xcolor}
\mathtoolsset{showonlyrefs}
\usepackage{bbm}
\usepackage{tikz}
\usepackage{caption}
\usepackage{pgfplots}

\usepackage{nicefrac}
\usepackage{jn_colors_parula}

\pgfplotsset{compat=newest,compat/show suggested version=false}

\allowdisplaybreaks
\DeclareCaptionLabelSeparator{periodspace}{.\quad}
\newcommand{\ed}{\textcolor{black}}
\newcommand{\own}{\textcolor{black}}
\newcommand{\edd}{\textcolor{black}}

\DeclareMathOperator{\rank}{rank}

\IEEEoverridecommandlockouts
\begin{document}
	\input{tex/notation}
	\title{An Information-Theoretic Perspective on Successive Cancellation List Decoding and Polar Code Design}
	\author{Mustafa Cemil Co\c{s}kun, \IEEEmembership{Student Member, IEEE}, Henry D. Pfister, \IEEEmembership{Senior Member, IEEE}
	\thanks{This paper was presented in part at the IEEE International Conference on Signal Processing and Communications, July 2020, Bangalore, India \cite{CP20}.}
	\thanks{Mustafa Cemil Co\c{s}kun is with the Institute for Communications Engineering (LNT), Technical University of Munich, Munich, Germany (email: mustafa.coskun@tum.de). Parts of this work were carried out when he was also with the Institute of Communications and Navigation of the German Aerospace Center (DLR), We{\ss}ling, Germany, and with the Department of Electrical and Computer Engineering, Duke University, Durham, USA.}
	\thanks{Henry D. Pfister is with the Department of Electrical and Computer Engineering, Duke University, Durham, USA (email: henry.pfister@duke.edu).}
	\thanks{The work of M. C. Co\c{s}kun was supported in part by the Helmholtz Gemeinschaft through the HGF-Allianz DLR@Uni project Munich Aerospace via the research grant ``Efficient Coding and Modulation for Satellite Links with Severe Delay Constraints'' and in part by the German Research Foundation (DFG) under Grant {KR 3517/9-1}. The work of H. D. Pfister was supported in part by the National Science Foundation under Grant No. 1718494. Any opinions, findings, conclusions, and recommendations expressed in this material are those of the authors and do not necessarily reflect the views of these sponsors.}
	}

	\maketitle

	\begin{abstract}
		\edd{This work identifies information-theoretic quantities that are closely related to the required list size on average for successive cancellation list (SCL) decoding to implement maximum-likelihood decoding \ed{over general binary memoryless symmetric (BMS) channels}}. It also provides \ed{upper and lower bounds} for these quantities that can be computed efficiently for very long codes. \ed{For the binary erasure channel (BEC), we} provide a simple method to estimate the mean \ed{accurately} via density evolution. \edd{The analysis shows how to modify, e.g., Reed-Muller codes,} to improve the performance when practical list sizes, e.g., $L\in\edd{[8, 1024]}$, are adopted. \ed{Exemplary constructions with block lengths $N\in\{128,512\}$ outperform polar codes of 5G over the binary-input additive white Gaussian noise channel.}
		
		\ed{It is further shown that t}here is a concentration around the mean of the logarithm of the required list size for sufficiently large block lengths, \ed{over discrete-output BMS channels. We provide the probability mass functions (p.m.f.s) of this logarithm, over the BEC, for a sequence of the modified RM codes with an increasing block length via simulations, which illustrate that the p.m.f.s concentrate around the estimated mean.}
	\end{abstract}
	\markboth
	{submitted to IEEE Transactions on Information Theory}
	{}
	\input{tex/introduction.tex}
	\input{tex/preliminaries.tex}
	\input{tex/analysis.tex}
	\input{tex/concentration.tex}
	\input{tex/numerical_results.tex}
	\input{tex/Conclusions.tex}
	\input{tex/appendix.tex}
	
	\section*{Acknowledgement}
	    The authors thank Gerhard Kramer (TUM)\own{, the associate editor and anonymous reviewers} for comments improving the presentation.
	\input{ms.bbl}

\end{document}

%% file: tex/notation.tex
\newcommand{\ensemble}{\mathscr{C}}
\newcommand{\code}{\mathcal{C}}
\newcommand{\vecu}{\boldsymbol{u}}
\newcommand{\veci}{\boldsymbol{i}}
\newcommand{\vecf}{\boldsymbol{f}}
\newcommand{\vecv}{\boldsymbol{v}}
\newcommand{\vecp}{\boldsymbol{p}}
\newcommand{\vecx}{\boldsymbol{x}}
\newcommand{\vecone}{\boldsymbol{1}}
\newcommand{\vecuhat}{\hat{\boldsymbol{u}}}
\newcommand{\vecc}{\boldsymbol{c}}
\newcommand{\vecb}{\boldsymbol{b}}
\newcommand{\vecchat}{\hat{\boldsymbol{c}}}
\newcommand{\vecy}{\boldsymbol{y}}
\newcommand{\vecz}{\boldsymbol{z}}
\newcommand{\F}{\boldsymbol{F}}
\newcommand{\B}{\boldsymbol{B}}
\newcommand{\G}{\boldsymbol{G}}
\newcommand{\V}{\boldsymbol{V}}
\newcommand{\GK}{\boldsymbol{K}}
\newcommand{\Gsys}{\G_{\mathsf{i},\mathsf{sys}}}
\newcommand{\Gnsys}{\G_{\mathsf{i},\mathsf{nsys}}}
\newcommand{\Per}{\boldsymbol{\Pi}}
\newcommand{\I}{\boldsymbol{I}}
\newcommand{\perP}{\boldsymbol{P}}
\newcommand{\perS}{\boldsymbol{S}}
\newcommand{\GH}[1]{\bm{\mathsf{G}}_{#1}}
\newcommand{\T}[1]{\bm{\mathsf{T}}{(#1)}}

\newcommand{\mi}{\mathrm{I}}
\newcommand{\Prob}{P}
\newcommand{\Z}{\mathrm{Z}}
\newcommand{\SPC}{\mathcal{S}}
\newcommand{\Rep}{\mathcal{R}}
\newcommand{\f}{\mathrm{f}}
\newcommand{\W}{\mathrm{W}}
\newcommand{\A}{\mathrm{A}}
\newcommand{\SC}{\mathrm{SC}}
\newcommand{\Q}{\mathrm{Q}}
\newcommand{\MAP}{\mathrm{MAP}}
\newcommand{\E}{\mathbb{E}}
\newcommand{\ent}{\mathrm{H}}
\newcommand{\U}{\mathrm{U}}
\newcommand{\Ham}{\mathrm{H}}

\newcommand{\Amin}{\mathrm{A}_{\mathrm{min}}}
\newcommand{\dmin}{d}
\newcommand{\erfc}{\mathrm{erfc}}
\newcommand{\inner}[1]{\langle#1\rangle}

\newcommand{\de}{\mathrm{d}}

\newcommand{\decoRule}{\rule{\textwidth}{.4pt}}

\newcommand{\oleq}[1]{\overset{\text{(#1)}}{\leq}}
\newcommand{\oeq}[1]{\overset{\text{(#1)}}{=}}
\newcommand{\ogeq}[1]{\overset{\text{(#1)}}{\geq}}
\newcommand{\ogeql}[2]{\overset{#1}{\underset{#2}{\gtreqless}}}
\newcommand{\argmax}[1]{\underset{#1}{\mathrm{arg \, max}}\,}
\newcommand\numeq[1]%
{\stackrel{\scriptscriptstyle(\mkern-1.5mu#1\mkern-1.5mu)}{=}}
\newcommand\numleq[1]%
{\stackrel{\scriptscriptstyle(\mkern-1.5mu#1\mkern-1.5mu)}{\leq}}

\newcommand{\pscd}{\gl{\Prob(\mathcal{E}_{\SCD})}}
\newcommand{\uED}{\gl{\hat{u}}}
\newcommand{\LED}{\gl{L_i^{(\ED)}}}
\newcommand{\inverson}[1]{\gl{\mathbb{I}\left\{#1\right\}}}

\newtheorem{mydef}{Definition}
\newtheorem{prop}{Proposition}
\newtheorem{theorem}{Theorem}

\newtheorem{proposition}{Proposition}
\newtheorem{lemma}{Lemma}
\newtheorem{remark}{Remark}
\newtheorem{example}{Example}
\newtheorem{definition}{Definition}
\newtheorem{corollary}{Corollary}


\definecolor{lightblue}{rgb}{0,.5,1}
\definecolor{normemph}{rgb}{0,.2,0.6}
\definecolor{supremph}{rgb}{0.6,.2,0.1}
\definecolor{lightpurple}{rgb}{.6,.4,1}
\definecolor{gold}{rgb}{.6,.5,0}
\definecolor{orange}{rgb}{1,0.4,0}
\definecolor{hotpink}{rgb}{1,0,0.5}
\definecolor{newcolor2}{rgb}{.5,.3,.5}
\definecolor{newcolor}{rgb}{0,.3,1}
\definecolor{newcolor3}{rgb}{1,0,.35}
\definecolor{darkgreen1}{rgb}{0, .35, 0}
\definecolor{darkgreen}{rgb}{0, .6, 0}
\definecolor{darkred}{rgb}{.75,0,0}
\definecolor{midgray}{rgb}{.8,0.8,0.8}
\definecolor{darkblue}{rgb}{0,.25,0.6}

\definecolor{lightred}{rgb}{1,0.9,0.9}
\definecolor{lightblue}{rgb}{0.9,0,0.0}
\definecolor{lightpurple}{rgb}{.6,.4,1}
\definecolor{gold}{rgb}{.6,.5,0}
\definecolor{orange}{rgb}{1,0.4,0}
\definecolor{hotpink}{rgb}{1,0,0.5}
\definecolor{darkgreen}{rgb}{0, .6, 0}
\definecolor{darkred}{rgb}{.75,0,0}
\definecolor{darkblue}{rgb}{0,0,0.6}

\definecolor{bgblue}{RGB}{245,243,253}
\definecolor{ttblue}{RGB}{91,194,224}

\definecolor{dark_red}{RGB}{150,0,0}
\definecolor{dark_green}{RGB}{0,150,0}
\definecolor{dark_blue}{RGB}{0,0,150}
\definecolor{dark_pink}{RGB}{80,120,90}

\begin{acronym}
	\acro{APP}{a-posteriori probability}
	\acro{AWGN}{additive white Gaussian noise}
	\acro{BIAWGN}{binary-input additive white Gaussian noise}
	\acro{B-DMC}{binary-input discrete memoryless channel}
	\acro{B-DMSC}{binary-input discrete memoryless symmetric channel}
	\acro{BMS}{binary memoryless symmetric}
	\acro{BCJR}{Bahl, Cocke, Jelinek, and Raviv}
	\acro{BEC}{binary erasure channel}
	\acro{BER}{bit error rate}
	\acro{BLEP}{block error probability}
	\acro{BLER}{block error rate}
	\acro{BP}{belief propagation}
	\acro{BSC}{binary symmetric channel}
	\acro{CER}{codeword error rate}
	\acro{CN}{check node}
	\acro{CRC}{cyclic redundancy check}
	\acro{DE}{density evolution}
	\acro{eBCH}{extended Bose-Chaudhuri-Hocquengham}
	\acro{FER}{frame error rate}
	\acro{GA}{Gaussian approximation}
	\acro{i.i.d.}{independent and identically distributed}
	\acro{IO-WE}{input-output weight enumerator}
	\acro{IR-WE}{input-redundancy weight enumerator}
	\acro{IO-WEF}{input-output weight enumerating function}
	\acro{IR-WEF}{input-redundancy weight enumerating function}
	\acro{JIO-WE}{joint IO-WE}
	\acro{JIR-WE}{joint IR-WE}
	\acro{JWE}{joint WE}
	\acro{LDPC}{low-density parity-check}
	\acro{LHS}{left-hand side}
	\acro{LLR}{log-likelihood ratio}
	\acro{MAP}{maximum a-posteriori}
	\acro{MC}{metaconverse}
	\acro{ML}{maximum-likelihood}
	\acro{NA}{normal approximation}
	\acro{PC}{product code}
	\acro{PW}{polarization weight}
	\acro{pdf}{probability density function}
	\acro{RCB}{random coding bound}
	\acro{RCU}{random coding union}
	\acro{RM}{Reed--Muller}
	\acro{dRM}{dynamic Reed--Muller}
	\acro{RHS}{right-hand side}
	\acro{r.v.}{Random variable}
	\acro{p.m.f.}{probability mass function}
	\acro{SPC}{single parity-check}
	\acro{SC}{successive cancellation}
	\acro{SCC}{super component codes}
	\acro{SCL}{successive cancellation list}
	\acro{SCI}{Successive cancellation inactivation}
	\acro{SISO}{soft-input soft-output}
	\acro{SNR}{signal-to-noise ratio}
	\acro{UB}{union bound}
	\acro{TUB}{truncated union bound}
	\acro{VN}{variable node}
	\acro{WE}{weight enumerator}
	\acro{WEF}{weight enumerating function}
	\acro{PAC}{polarization-adjusted convolutional}
\end{acronym}

%% file: tex/introduction.tex
\section{Introduction}\label{sec:intro}

Polar codes constitute the first deterministic construction of capacity-achieving codes for \ac{BMS} channels with an efficient decoder\cite{arikan2009channel}. While they achieve capacity under \ac{SC} decoding, their initial performance was not competitive with \ac{LDPC} and Turbo codes. This changed with the advent of \ac{SCL} decoding and the addition of \ac{CRC} outer codes\cite{tal15}. Due to their competitive performance for short block lengths\cite{Coskun18:Survey}, they have been adopted by the 5G standard\cite{5G20}. Many authors have optimized polar codes and their variants to improve their performance for SCL decod\ed{ing \cite{trifonov16,Trifonov17,Yuan19,Elkelesh19,Fazeli19,Rowshan19,arikan2019sequential,MV20,Trifonov20}}.

An important property of SCL decod\ed{ing} is that its performance matches the \ac{ML} decoder if the list size is sufficiently large. \ed{This work is motivated by} the theoretical question: ``What list size suffices to \ed{approach} ML decoding performance for a given channel quality?''. Simulating SCL decoding with a large list size is infeasible for long codes and doesn't provide much insight \ed{into designing codes for \ac{SCL} decoding}. Current works rely mostly on heuristics, e.g., see \cite{Trifonov17,Yuan19,Elkelesh19,Rowshan19,MV20}.

By noting that \ac{SCL} decoding operates sequentially in $N$ stages, where $N$ is the block length, \edd{we identify information-theoretic quantities associated with the required list size on average.} For general \ac{BMS} channels, we provide upper and lower bounds that can be computed efficiently even for very long codes. Our analysis suggests new code design \edd{criterion} for \ac{SCL} \ed{decoding} with practical list sizes, e.g., $L\in\edd{[8, 1024]}$ \edd{by modifying, e.g., \ac{RM} codes}\cite{reed,muller}. This is illustrated via \own{exemplary} constructions for short- to moderate-length regime, e.g., $N\in\{128,512\}$\ed{, which outperform polar codes of the 5G standard}. For the \ac{BEC}, we provide a simple Markov chain approximation to compute \ed{the} mean using only density evolution and simulations show that it is quite accurate.

\ed{Motivated by the numerical results of \cite{CP20}, we show that, for discrete-output \ac{BMS} channels,} the logarithm of required list size \ed{at each decoding stage concentrates around the mean for large block lengths using techniques similar to those for analyzing LDPC codes \cite{LMSS98,RU01}. Over the \ac{BEC}, we provide its \ac{p.m.f.} for modified \ac{RM} codes with several block lengths. Simulations illustrate the fact that, with increasing block lengths, the \acp{p.m.f.} concentrate around the estimated mean.}

The paper is organized as follows. In Section \ref{sec:prelim}, we provide the preliminaries needed for the rest of the work. In Section \ref{sec:analysis}, we introduce the key quantities and use them to analyze the considered list decoders. The convergence properties are studied for general \ac{BMS} channels as well as for the BEC in Section \ref{sec:concentration}. Numerical results are presented in Section \ref{sec:numerical} for the \ac{BIAWGN} channel and the BEC. Conclusions follow in Section \ref{sec:conc}.

%% file: tex/preliminaries.tex
\section{Background}\label{sec:prelim}

\acp{r.v.} are denoted by upper case letters, e.g., $X$, and their realizations by the lower case counterparts, e.g., $x$. Random vectors are denoted by $X_{i}^{j}=(X_{i},X_{i+1},\ldots,X_{j})$ and their realizations by $x_{i}^{j}$. If $j<i$, then $X_{i}^{j}$ is void. We use $[N]$ for the set $\left\{ 1,2,\ldots N\right\}$. Subvectors with indices in $\mathcal{S}\subseteq[N]$ are denoted by $x_{\mathcal{S}}=(x_{i_{1}},\ldots,x_{i_{|\mathcal{S}|}})$ where $i_{1}<\cdots<i_{|\mathcal{S}|}$ enumerates the elements in $\mathcal{S}$ with $|\mathcal{S}|$ being the cardinality of the set $\mathcal{S}$. We use $x_{\sim i}$ for the vector $x_{[N]\setminus\{i\}}$. 
Bold capital letters are used for matrices, e.g., $\boldsymbol{X}$.

Consider a BMS channel with binary input $X\in\left\{ 0,1\right\}$ and general output $Y\in\mathbb{\mathcal{Y}}$, i.e., $W:X\rightarrow Y$. The transition probabilities are given by $W(y|x)\triangleq\Pr(Y=y|X=x)$.

\subsection{Polar and Reed-Muller Codes}

\label{sec:polar_RM}

The polar transform of length $N=2^n$ is denoted by \ed{$\G_{N}\triangleq \B_N \G_{2}^{\otimes n}$, where $\B_N$ is the $N\times N$ bit-reversal matrix~\cite[Sec. VII.B]{arikan2009channel} and $\G_{2}^{\otimes n}$} is the $n$-fold Kronecker product of the $2\times2$ binary Hadamard matrix
\begin{equation}\label{eq:building_block}
	\G_{2}\triangleq\left[\begin{array}{cc}
	1 & 0\\
	1 & 1
	\end{array}\right].
\end{equation}
This is the key building block in Ar{\i}kan's polar codes\cite{arikan2009channel} and \ac{RM} codes\cite{reed,muller}.

To define an $(N,K)$ polar or \ac{RM} code, one partitions the input vector into sub-vectors that carry information and frozen bits whose values are known by the receiver. The set of information and frozen indices are denoted, respectively, by $\mathcal{A}\subseteq[N]$ with $|\mathcal{A}| = K$, and $\mathcal{F}\triangleq[N]\setminus\mathcal{A}$. Thus, the input vector $u_{1}^{N}$ can be split into information bits $u_{\mathcal{A}}$ and frozen bits $u_{\mathcal{F}}$. Then, the codeword $x=u\G_{N}$ is transmitted over the channel. This construction enables efficient SC decoding \cite{arikan2009channel,stolte2002rekursive}. For polar codes, the set $\mathcal{A}$ is chosen to minimize a tight upper bound on the error probability under SC decoding; however, the set $\mathcal{A}$ for an $r$-th order \ed{\ac{RM} code of length $N$}, denoted by RM$(r,\edd{n})$, consists of the indices of rows in $\G_{N}$ with the Hamming weight at least equal to $2^{\edd{n}-r}$.

\subsection{Successive Cancellation Decoding}

Let $y_{1}^{N}$ be observations of the bits $x_{1}^{N}$ through $N$ copies of the BMS channel $W$. SC decod\ed{ing} takes the following steps sequentially from $i=1$ to $i=N$. If $i\in\mathcal{F}$, it sets $\hat{u}_{i}$ to its frozen value. If $i\in\mathcal{A}$, it computes the soft estimate $p_{i}(\hat{u}_1^{i-1})\triangleq\Pr(U_{i}=1|Y_{1}^{N}=y_{1}^{N},U_1^{i-1}=\hat{u}_1^{i-1})$, and makes a hard decision accordingly as
\begin{equation}\label{eq:dec_fnc}
	\hat{u}_{i}=\begin{cases}
	0 & \text{if }p_{i}(\hat{u}_1^{i-1})<\frac{1}{2}\\
	1 & \text{otherwise}.
	\end{cases}
\end{equation}

To understand SC decod\ed{ing}, we focus on the effective channels seen by each of the input bits in $u_{1}^{N}$\cite{arikan2009channel}. SC decod\ed{ing} uses the entire $y_{1}^{N}$ vector and all past decisions $\hat{u}_{1}^{i-1}$ to generate the soft estimate $p_{i}(\hat{u}_1^{i-1})$ and the hard decision $\hat{u}_{i}$ for $u_{i}$. Let $W_{N}^{(i)}$ denote the effective (virtual) channel seen by $u_{i}$ during SC decoding\cite{arikan2009channel}. If all past bits $u_{1}^{i-1}$ are provided by a genie, then this channel is easier to analyze. The effective channel $W_{N}^{(i)}:U_{i}\rightarrow(Y_{1}^{N},U_{1}^{i-1})$ is then defined by its transition probabilities
\begin{equation}\label{eq:effective_channel}
	W_{N}^{(i)}\left(y_{1}^{N},u_{1}^{i-1}|u_{i}\right)\triangleq\!\!\sum_{u_{i+1}^{N}\in\{0,1\}^{N-i}}\frac{1}{2^{N-1}}W_{N}(y_{1}^{N}|u_{1}^{N}\G_{N})
\end{equation}
where $W_{N}(y_{1}^{N}|x_{1}^{N})\triangleq\prod_{i=1}^{N}W(y_{i}|x_{i})$.

\subsection{Successive Cancellation List Decoding}\label{sec:SCL_decoding}

SCL decoding of RM codes (and related subcodes) was introduced in\cite{Dumer01}. \ed{It was then} extended to optimized constructions of generalized concatenated codes in \cite{stolte2002rekursive}. These approaches became popular after \cite{tal15} applied them to polar codes that were combined with an outer CRC code to \ed{improve their finite-length performance}.

SCL decoding recursively computes $Q_{i}(\tilde{u}_{1}^{i},y_1^N)$\own{, which is proportional to} $\Pr\left(U_{1}^{i}=\tilde{u}_{1}^{i},Y_{1}^{N}=y_{1}^{N}\right)$ for $i=1,\dots,N$ via the SC message passing rules for partial input sequences $\tilde{u}_{1}^{i}\in \mathcal{U}_i \subseteq \left\{ 0,1\right\} ^{i}$\own{, which are also called decoding paths.} Discarding $y_1^N$ for the ease of notation, we refer to $Q_{i}(\tilde{u}_{1}^{i})$ as the \emph{myopic likelihood} of the sequence $\tilde{u}_{1}^{i}$ as it does not use the receiver's knowledge of frozen bits after $u_{i}$.

Let $\mathcal{U}_{i-1}\subseteq\left\{ 0,1\right\}^{i-1}$ be a subset satisfying $\left|\mathcal{\mathcal{U}}_{i-1}\right|=L$ and assume that $Q_{i-1}(\tilde{u}_{1}^{i-1})$ is known for some $\tilde{u}_{1}^{i-1}\in\mathcal{\mathcal{U}}_{i-1}$. Then, for $\tilde{u}_{i}\in\{0,1\}$, one can write
\begin{align}\label{eq:recursive_scl}
&\!\!\!\!\!\!\!\!Q_{i}(\tilde{u}_{1}^{i})\propto\Pr\left(U_{1}^{i}=\tilde{u}_{1}^{i},Y_{1}^{N}=y_{1}^{N}\right) \nonumber\\
&\propto Q_{i-1}(\tilde{u}_{1}^{i-1})\Pr\left(U_{i}=\tilde{u}_{i}|Y_{1}^{N}=y_{1}^{N},U_{1}^{i-1}=\tilde{u}_{1}^{i-1}\right)
\end{align}
where the right-most term can be computed efficiently by the SC decoder starting from $Q_{0}(\tilde{u}_{1}^{0})\triangleq 1$. This results in $Q_{i}(\tilde{u}_{1}^{i})$ values for $2L$ partial sequences \ed{produced from $\mathcal{U}_{i-1}$}. One then prunes the list down to $L$ sequences by keeping only most likely paths according to \eqref{eq:recursive_scl} for SCL \ed{decoding} with list size $L$. Note that if $u_{i}$ is frozen, then the decoder simply extends all paths with correct frozen bit. After the $N$-th decoding stage, the estimate $\hat{u}_{1}^{N}$ is chosen as the candidate maximizing the function $Q_{N}(\tilde{u}_{1}^{N})$.

For the BEC, all partial input sequences $u_{1}^{i}\in\{0,1\}^{i}$ with $Q_{i}(u_{1}^{i})>0$ are equiprobable. Hence, for a given $u_{1}^{i-1}$, we have $p_{i}(u_{1}^{i-1})\in\left\{ 0,\nicefrac{1}{2},1\right\}$. If $p_{i}(u_{1}^{i-1})\in\left\{ 0,1\right\}$, then $u_{i}$ is known perfectly at the receiver. However, if $p_{i}(u_{1}^{i-1})=\nicefrac{1}{2}$ and $u_{i}$ is not frozen, then $\hat{u}_{i}$ can be seen as an erasure. When an information bit is decoded to an erasure, it is replaced by both possible values (assuming the list size is large enough) and decoding proceeds separately under these two hypotheses. The number of hypotheses (i.e., partial input sequences) is halved if the decoded value \edd{for a frozen bit contradicts with its actual value for different partial input sequences}\cite[Appendix A]{Neu2019}. An error is declared if, at any decoding stage, the number of hypotheses exceeds $L$ or if there are more than one hypothesis at \ed{the end of} the process, i.e., if $\exists\tilde{u}_{1}^{N}\in\{0,1\}^{N}$, $\tilde{u}_{1}^{N}\neq u_1^N$ with $Q_{N}(\tilde{u}_{1}^{N})>0$.

\subsection{Successive Cancellation Inactivation Decoding}

\own{\ac{SCI}} decoding is a more efficient version of SCL decoding for the BEC proposed in \cite{CNP20}. It follows the same decoding schedule as SCL decod\ed{ing} but instead replaces an erased information bit by \own{a variable}, i.e., the bit is \emph{inactivated} \cite{Shokrollahi06,Measson08,Eslami10}. Later, \edd{the decoding function for any frozen bit provides a linear equation if its output is not an erasure \cite[Eq. (3)]{CNP20}. Thanks to this equation, the} inactivated bit may be resolved. This can be done at \edd{any} stage of the decoding process \edd{when the obtained equation is not a trivial one} or delayed until the end. When \edd{a dummy variable is} eliminated without delaying to the end of decoding, we refer to this as a \emph{consolidation} event. In this work, our focus is on SCI \ed{decoding} with consolidations.

SCI decod\ed{ing} can inactivate multiple bits if required. If the maximum number of inactivations is not bounded, then it implements ML decoding. For further details, see \cite{CNP20}. 

\subsection{Dynamic Frozen Bits}
\label{sec:polar_ebch}

An important observation in \cite{trifonov16} is that decoders having the successive nature (e.g., SCL \own{decoding}) still work (with slight modifications) \ed{even} if, for some $i\in\mathcal{F}$, the bit $u_{i}$ is a function of a set of preceding information bits. A frozen bit whose value depends on past inputs is called dynamic. 

A polar code with dynamic frozen bits is defined by its information indices $\mathcal{A}$ and a matrix that defines each frozen bit as a linear combination of preceding information bits. There are now a number of approaches for choosing $\mathcal{A}$ and dynamic frozen bit constraints\edd{, e.g., \cite{trifonov16,Trifonov17,Yuan19,Trifonov20,MV20}, which aim at having a resulting code with large minimum distance and/or small multiplicity of minimum weight codewords as well as low to moderate decoding complexity for target \acp{BLER}.} In this work, after specifying $\mathcal{A}$, we define each frozen bit to be a uniform random linear combination of information bits preceding it.

%% file: tex/analysis.tex
\section{Analysis of the List Decoders}\label{sec:analysis}

An important property of list decoding is that, if the correct codeword is on the list at the end of decoding, then the error probability is upper bounded by that of ML decod\ed{ing}. We study how large the list should be at each stage so that the correct codeword is likely to be on the list.

\subsection{An Information-Theoretic Perspective}

Consider a length-$N$ polar code with SCL decoding after the $m$-th decoding stage. Since SCL decoding does not use future frozen bits, we focus on the subset of length-$m$ input \ed{sequences} that have significant conditional entropy given the channel observation. An important insight is that, after observing $Y_{1}^{N}$, the uncertainty in $U_{1}^{m}$ is quantified by the entropy
\begin{equation}\label{eq:conditional_entropy}
	H\left(U_{1}^{m}|Y_{1}^{N}\right)=\sum_{i=1}^{m}H\left(U_{i}|U_{1}^{i-1},Y_{1}^{N}\right)
\end{equation}
where $U_{1}^{N}$ is assumed to be uniform over $\left\{0,1\right\}^{N}$. This is exactly true if the first $m$ bits are all information bits, i.e., if $[m]\subseteq\mathcal{A}$. If $[m]$ contains also frozen indices, however, then the situation is more complicated.\edd{\footnote{\edd{Here, it is assumed that the frozen bits $U_{\mathcal{F}}$ are also uniformly distributed and SCL decoding learns them causally. In the case of dynamic frozen bits, uniform random constants are added, which are revealed causally as in the case of other frozen bits.}}}

Let $\mathcal{A}^{(m)}\triangleq\mathcal{\mathcal{A}}\cap[m]$ and $\mathcal{F}^{(m)}\triangleq\mathcal{F}\cap[m]$ be the sets containing information and frozen indices within the first $m$ input bits, respectively. Now consider an experiment where the frozen bits $U_i$ with $i\in\mathcal{F}^{(m)}$ are uniform and independent of $U_{1}^{i-1}$. Using \eqref{eq:conditional_entropy} naively with the assumption that $U_{\mathcal{F}^{(m)}}$ is not known to the receiver would cause an overestimate of $H\left(U_{1}^{m}|Y_{1}^{N}\right)$ by an amount of at least $\sum_{i\in\mathcal{F}^{(m)}}H\left(U_{i}|U_{1}^{i-1},Y_{1}^{N}\right)$. In addition to this, the frozen bits $U_{\mathcal{F}^{(m)}}$ may reveal additional information about the previous information bits.

To better understand the uncertainty within the first $m$ input bits during SCL decoding, we define the \ed{quantity
\begin{equation}\label{def:cond_entropy}
    d_{m}\left(y_1^N\right)\triangleq H\left(U_{\mathcal{A}^{(m)}}|Y_{1}^{N}=y_1^N,U_{\mathcal{F}^{(m)}}\right)
\end{equation}
and the corresponding \ac{r.v.} \edd{is denoted as $D_{m}$ that takes on the value $d_m\left(y_1^N\right)$ when $Y_1^N = y_1^N$.}
Note that the mean of $D_m$ corresponds to the conditional entropy, i.e.,
\begin{equation}
	\Bar{D}_m \triangleq \E\left[d_{m}\left(Y_1^N\right)\right] = H\left(U_{\mathcal{A}^{(m)}}|Y_{1}^{N},U_{\mathcal{F}^{(m)}}\right)
\end{equation}
and we define a difference sequence as
\begin{equation}
	\Delta_m \triangleq \Bar{D}_m - \Bar{D}_{m-1}.
\end{equation}}
Observe that, if $U_{m}$ is an information bit, then we have
\begin{align}\label{eq:info_bit_analysis}
	\Delta&_{m}=\!H\left(U_{\mathcal{A}^{(m)}}|Y_{1}^{N},U_{\mathcal{F}^{(m)}}\right)\!-\!H\left(U_{\mathcal{A}^{(m-1)}}|Y_{1}^{N},U_{\mathcal{F}^{(m-1)}}\right)\nonumber\\
	&\!=\!H\left(U_{\mathcal{A}^{(m)}}|Y_{1}^{N},U_{\mathcal{F}^{(m-1)}}\right)-H\left(U_{\mathcal{A}^{(m-1)}}|Y_{1}^{N},U_{\mathcal{F}^{(m-1)}}\right)\nonumber\\
	&\!=\!H\left(U_{\mathcal{A}^{(m)}},U_{\mathcal{F}^{(m-1)}}|Y_{1}^{N}\right)-H\left(U_{\mathcal{A}^{(m-1)}},U_{\mathcal{F}^{(m-1)}}|Y_{1}^{N}\right)\nonumber\\
	&\!=\!H\left(U_{1}^{m-1}|Y_{1}^{N}\right)+H\left(U_{m}|Y_{1}^{N},U_{1}^{m-1}\right)\!-\!H\left(U_{1}^{m-1}|Y_{1}^{N}\right)\nonumber\\
	&\!=\!H(U_{m}|Y_{1}^{N},U_1^{m-1})
\end{align}
which is exactly what one would expect from the naive analysis given by \eqref{eq:conditional_entropy}.

If $U_{m}$ is a frozen bit, then consider a model where it is not known to the receiver at the time of transmission.\footnote{This reflects how SCL decod\ed{ing} operates, i.e., it does not use the knowledge of any frozen bit $U_{m}$ until reaching the end of its decoding stage $m$. The soft estimate $p_{m}(u_{1}^{m-1})$ provides an additional information to separate the hypotheses (i.e, paths) although the hard estimate is chosen as $\hat{u}_{m}=u_{m}$ independent of $p_{m}(u_{1}^{m-1})$.} The act of revealing $U_{m}$ to the receiver changes the conditional uncertainty about $U_{\mathcal{A}^{(m-1)}}$ by
\begin{align}\label{eq:one_minus_entropy}
	&\Delta_{m}=H\left(U_{\mathcal{A}^{(m)}}|Y_{1}^{N},U_{\mathcal{F}^{(m)}}\right)-H\left(U_{\mathcal{A}^{(m-1)}}|Y_{1}^{N},U_{\mathcal{F}^{(m-1)}}\right) \nonumber\\
	&=\!H\!\left(U_{\mathcal{A}^{(m-1)}}|Y_{1}^{N},U_{\mathcal{F}^{(m-1)}},U_{m}\!\right)\!-\!H\!\left(U_{\mathcal{A}^{(m-1)}}|Y_{1}^{N},U_{\mathcal{F}^{(m-1)}}\!\right) \nonumber\\
	&=-I\left(U_{m};U_{\mathcal{A}^{(m-1)}}|Y_{1}^{N},U_{\mathcal{F}^{(m-1)}}\right) \nonumber\\
	&=H\left(U_{m}|Y_{1}^{N},U_{1}^{m-1}\right)-H\left(U_{m}|Y_{1}^{N},U_{\mathcal{F}^{(m-1)}}\right) \nonumber\\
	&\geq H\left(U_{m}|Y_{1}^{N},U_{1}^{m-1}\right)-1.
\end{align}
This expression quantifies the effect of revealing the new frozen bit as a reduction in the conditional entropy of the information bits preceding it. A large reduction may occur when the channel $W_{N}^{(m)}$ has low entropy (i.e., a low-entropy effective channel is essentially frozen) and the reduction will be small if the channel entropy is high (i.e., the input is unpredictable from $Y_1^N$ and $U_1^{m-1}$).

For BMS channels, we can combine \eqref{eq:info_bit_analysis} and \eqref{eq:one_minus_entropy} to understand the dynamics of \ed{$\Bar{D}_m$}. This gives a proxy for the uncertainty in SCL decoding after $m$ steps. We have
\begin{subequations}
	\begin{align}
	\sum_{i\in\mathcal{A}^{(m)}}H\left(W_{N}^{(i)}\right)-\sum_{i\in\mathcal{F}^{(m)}}&\left(1-H\left(W_{N}^{(i)}\right)\right)\leq \ed{\Bar{D}_m} \label{eq:bounds_LB}\\
	&\quad\leq\sum_{i\in\mathcal{A}^{(m)}}H\left(W_{N}^{(i)}\right).\label{eq:bounds_UB}
	\end{align}
\end{subequations}
The lower bound assumes that frozen bits (when perfectly observed) always reduce the entropy.

\begin{theorem}
	Upon observing $y_1^N$ when $u_1^N$ is transmitted, the set of partial sequences $\tilde{u}_1^m$ with a larger likelihood than some fraction of that for true sequence $u_1^m$ after $m$ stages of SCL decoding is given by \ed{\begin{equation}
	    \mathcal{S}_\alpha^{(m)}\left(u_1^m,y_1^N\right)\triangleq\{\tilde{u}_1^m:Q_m(\tilde{u}_1^m)\geq \alpha Q_m(u_1^m)\}
	\end{equation}
	where the fraction is given by a positive number $\alpha\leq1$}. On average, the logarithm of its cardinality is upper bounded by
	\begin{align}
		\E\left[\log_{2}|\mathcal{S}_\alpha^{(m)}|\right]&\leq \ed{\Bar{D}_m} +\log_{2}\alpha^{-1}\label{eq:bound2inc} \\
		& = H\left(U_{\mathcal{A}^{(m)}}|Y_{1}^{N},U_{\mathcal{F}^{(m)}}\right)+\log_{2}\alpha^{-1}\label{eq:bound2}.
	\end{align}
\end{theorem}
\begin{proof}
	Assume, w.l.o.g., that $u_1^N$ and $y_1^N$ are transmitted and observed, respectively. Then, we have
	\begin{align}
		\log_{2}|\mathcal{S}_\alpha^{(m)}| \!&\numeq{\text{a}} \!\log_{2}\sum_{\tilde{u}_1^m}\!\mathbbm{1}_{\left(p(\tilde{u}_{\mathcal{A}^{(m)}}|y_1^N,u_{\mathcal{F}^{(m)}}\!)\geq \alpha\cdot p(u_{\mathcal{A}^{(m)}}|y_1^N,u_{\mathcal{F}^{(m)}}\!)\!\right)} \\[-1.25em]
		&\numleq{\text{b}} -\log_{2}\alpha\cdot p\left(u_{\mathcal{A}^{(m)}}|y_1^N,u_{\mathcal{F}^{(m)}}\right)
	\end{align}
	where ($\text{a}$) follows from $Q_m(u_1^m)\!\propto\! \Pr\left(U_{1}^{m}=\tilde{u}_{1}^{m},Y_{1}^{N}=y_{1}^{N}\right)$ and Bayes' rule and ($\text{b}$) from the fact that if there are more than $\left(\alpha\cdot p\left(u_{\mathcal{A}^{(m)}}|y_1^N,u_{\mathcal{F}^{(m)}}\right)\right)^{-1}$ sequences $\tilde{u}_{\mathcal{A}}$ with probability $\alpha\cdot p\left(u_{\mathcal{A}^{(m)}}|y_1^N,u_{\mathcal{F}^{(m)}}\right)$, then the total probability exceeds $1$. As the inequality is valid for any pair $u_1^N$ and $y_1^N$, taking the expectation over all $u_1^m$ and $y_1^N$ yields the stated result.
\end{proof}
\ed{Note that $\alpha$ is used as a tuning parameter to catch near misses. Making it too small will keep many partial sequences with low probabilities in the set while choosing it as $\alpha=1$ will exclude the decoding paths with probabilities close to the correct one. For Monte-Carlo simulations validating \eqref{eq:bound2}, $\alpha$ is chosen close to $1$, but still keeping $\E\left[\log_{2}|\mathcal{S}_\alpha^{(m)}|\right]$ close to $\bar{D}_m$ especially for small values of $m$.}

Now, consider SCL decod\ed{ing} whose list size is $L_{m}$ during the $m$-th decoding step. Then, the decoder should satisfy $L_{m} \geq |\mathcal{S}_1^{(m)}|$ for the true $u_1^m$ to be in the set $\mathcal{S}_1^{(m)}$.
Using \eqref{eq:bound2} and \eqref{eq:bounds_UB} yields the simple upper bound
\begin{equation}\label{eq:bound3}
	\E\left[\log_{2}|\mathcal{S}_\alpha^{(m)}|\right]\leq\sum_{i\in\mathcal{A}^{(m)}}H\left(W_{N}^{(i)}\right)+\log_{2}\alpha^{-1}.
\end{equation}
\begin{remark}\label{remark:rare_events}
	The analysis in terms of $\log_{2}L_{m}$ has two weaknesses. First, the entropy $\ed{\Bar{D}_m}$ only characterizes typical events, e.g., ensuring that the correct codeword stays on the list at least half of the time, whereas coding typically focuses on rarer events, e.g., \own{\acp{BLER}} less than $10^{-2}$. Second, the sequence $\ed{\Bar{D}_m}$ is averaged over $Y_1^{N}$ but the actual decoder sees a random realization $d_m(y_1^N) = H\left(U_{\mathcal{A}^{(m)}}|Y_1^{N}=y_1^{N},U_{\mathcal{F}^{(m)}}\right)$. Nevertheless, we believe the results provide a useful step towards a theoretical analysis of SCL \ed{decoding}. The numerical results in Section \ref{sec:numerical} illustrate the accuracy of the analysis. To motivate the analysis further, we study the convergence properties of the \ac{r.v.} $D_m$ in Section \ref{sec:concentration}.
\end{remark}
\begin{remark}\label{remark:design}
	The results have significance for code design. To achieve good performance \ed{under} SCL \ed{decoding} whose list size is $L_{m}$ during the $m$-th decoding step, a reasonable first-order design criterion is that $\log_{2}L_{m}\geq \ed{\Bar{D}_m}$. This observation implies, in principle, that frozen bits should be allocated to prevent $\ed{\Bar{D}_m}$ from exceeding $\log_{2}L_{m}$. \ed{Since computing $\Bar{D}_m$ requires simulations with huge list sizes (if not unbounded) and the upper bound \eqref{eq:bounds_UB} ignores the affect of frozen bits, we use the lower bound \eqref{eq:bounds_LB} as the proxy for designs.} \ed{Motivated by the observations of Remark \ref{remark:rare_events}, we believe that a relatively low \ac{SNR} values should be chosen for the analysis in order to capture rare events. Using these guidelines, e}xemplary designs are provided in Section \ref{sec:numerical_awgn}.
\end{remark}
\subsection{The Binary Erasure Channel}
\begin{proposition}\label{prop:subspace}
	On the BEC, the list of all valid partial \ed{input} sequences \ed{$u_{1}^m$} generated by SCL decoding with unbounded list size \ed{upon observing $y_1^N$} form an affine subspace\ed{, denoted as $\mathcal{S}^{(m)}\left(y_1^N\right)$}.
\end{proposition}
\begin{proof}
	Let $\mathcal{E}$ denote the set of erased positions in the realization $y_1^N$. Then, we can write
	\begin{equation}\label{eq:linear_system}
	(u_{\mathcal{A}^{(m)}},u_{m+1}^N)\G^\prime_{[N]\setminus\mathcal{F}^{(m)}}= y_{[N]\setminus\mathcal{E}}\oplus u_{\mathcal{F}^{(m)}}\G^\prime_{\mathcal{F}^{(m)}}
	\end{equation}
	where $\G^\prime_{\mathcal{S}}$ is the matrix formed by the rows of $\G_{N}$ indexed in $\mathcal{S}$ and then removing its columns indexed in $\mathcal{E}$ and $\oplus$ is the bit-wise XOR of two vectors. This equation enables to use the frozen bits $u_{\mathcal{F}^{(m)}}$ as side information. Let $\mathcal{C}$ denote the set of all possible solutions for $(u_{\mathcal{A}^{(m)}},u_{m+1}^N)$, which is an affine subspace. We are interested in all compatible partial information sequences $u_{\mathcal{A}^{(m)}}$ with \eqref{eq:linear_system} (hence, $u_1^m$ as $u_{\mathcal{F}^{(m)}}$ is a linear transform of $u_{\mathcal{A}^{(m)}}$). To this end, we define the mapping $\Pi_{\mathcal{A}^{(m)}}:\mathbb{F}_2^{N-|\mathcal{F}^{(m)}|}\rightarrow\mathbb{F}_2^{|\mathcal{A}^{(m)}|}$ as
	\begin{equation}\label{eq:linear_map}
	\Pi_{\mathcal{A}^{(m)}}(\mathcal{C})\triangleq\left\{v_1^{|\mathcal{A}^{(m)}|}:v_1^{N-|\mathcal{F}^{(m)}|}\in\mathcal{C}\right\}    
	\end{equation}
	which is a linear mapping since it can be represented as a multiplication of the input by a matrix formed by stacking an $|\mathcal{A}^{(m)}|\times |\mathcal{A}^{(m)}|$ identity matrix and an $(N-m)\times |\mathcal{A}^{(m)}|$ all-zero matrix. Then, the result follows by noting that a linear transform of an affine subspace is affine.
\end{proof}
\ed{Let $L_m\left(y_1^N\right)$ denote the list length after $m$-th decoding stage \edd{of SCL decoding with unbounded list size}, i.e., $L_m\left(y_1^N\right) =\left|\mathcal{S}^{(m)}\left(y_1^N\right)\right|$. Since each member of the list is equally likely, the proposition \edd{provides an} immediate \edd{corollary, which follows from the definition \eqref{def:cond_entropy} and observing that a subspace dimension is always a non-negative integer. The corollary bridges between the introduced quantity $d_m\left(y_1^N\right)$ and the list length $L_m\left(y_1^N\right)$ SCL decoding with unbounded list size explicitly.}
\begin{corollary}\label{cor:list_length_vs_sub_dim}
	At any decoding stage $m$, the list length $L_m\left(y_1^N\right)$, $y_1^N\in\{0,?,1\}^N$, of SCL decoding satisfies
	\begin{equation}
		\log_2L_{m}\left(y_1^N\right) = d_m\left(y_1^N\right). \label{eq:list_size_vs_sub_dim}
	\end{equation}
	Hence, \edd{$L_{m}\left(y_1^N\right)$} is a non-negative \edd{integer} power of $2$.
\end{corollary}
\edd{Next, another corollary is provided, which follows from the fact that $L_{m}\left(y_1^N\right)$ is a non-negative \edd{integer} power of $2$ in combination with the error events of SCL decoding with a fixed list size $L$ defined in Section \ref{sec:SCL_decoding}, i.e., either $\left|\mathcal{S}^{(m)}\left(y_1^N\right)\right| > L$ for any $m\in[N]$ or $\left|\mathcal{S}^{(N)}\left(y_1^N\right)\right| \neq 1$.}
\begin{corollary}
	SCL decoding with the list size $L$ performs the same as SCL decoding with $L^\prime = 2^{\lfloor\log_2 L\rfloor}$.
\end{corollary}}

\ed{The question of ``how large should the list size be for ML decoding?'' was partially addressed by \cite[Thm. 1]{HMH18} for the case of the BEC by providing an upper bound. We improve this bound with the following \edd{lemma}.
	\begin{lemma}\label{theorem:SCL_list_size_BEC}
		Let $L^*(\code)$ be the smallest list size for SCL decoding that implements ML decod\ed{ing} for an $(N,K)$ binary linear code $\code$. Let $\zeta$ and $\gamma$ denote the index of the last frozen bit before the first information bit, and the last (dynamic) frozen bit when $\code$ is represented as a polar code (with dynamic frozen bits), i.e.,
		\begin{equation*}
		    \zeta \triangleq \min\mathcal{A}\quad\text{and}\quad \gamma \triangleq \max\mathcal{F}.
		\end{equation*}
		Then $L^*(\code)$ is upper bounded as
		\begin{equation}\label{eq:SCL_list_size_BEC}
		    L^*(\code)\leq \min\left\{2^{N(1-R)-(\zeta-1)},2^{\gamma-N(1-R)}\right\}
		\end{equation}
		\edd{where $R$ is the code rate.}
	\end{lemma}}
\begin{proof}
	\ed{Observe that \edd{the first $\zeta-1$ frozen bits are of no use for pruning paths. Hence, there are $N(1-R)-(\zeta-1)$} frozen bits, which have the potential to prune half of the existing paths. If there are more than $2^{N(1-R)-(\zeta-1)}$ paths at any decoding stage during SCL decoding, then there will be at least $2$ solutions at the end.}
	
	\edd{We refer the reader to \cite[Thm. 1]{HMH18}, which states that $L^*(\code)\leq 2^{\gamma-N(1-R)}$, concluding the proof.}
\end{proof}
\ed{Note that \cite[Thm. 1]{HMH18} states that $L^*(\code)\leq 2^{\gamma-N(1-R)}$, which is usually relevant for low rate codes as $N(1-R)$ is large. Depending on the allocation of the frozen bit indices, \eqref{eq:SCL_list_size_BEC} can improve the previous result significantly. Consider, for instance, \ac{RM}$(5,7)$ with parameters $N=128$ and $K=120$ where $\zeta=4$ and $\gamma=65$. The previous result states that $L^*\leq 2^{57}$ while \eqref{eq:SCL_list_size_BEC} gives $L^*\leq 2^{5}$. Hence, \edd{Lemma} \ref{theorem:SCL_list_size_BEC} usually tightens the bound for high-rate codes. However, even this bound is far from being practical, especially for codes with rates $R\approx 0.5$. Consider, for instance, \ac{RM}$(3,7)$, where \eqref{eq:SCL_list_size_BEC} gives $L^*\leq 2^{49}$. In addition, the bounds are obviously independent of the channel quality since they are interested in exact ML decoding. If $\epsilon$ is very low then one requires much shorter lists on average to decode successfully.}

\ed{Recall that SCL decoding branches out the paths not for each information bit but whenever necessary. Instead, information bits are inactivated whenever necessary in the case of \ac{SCI} decoding. Therefore, we study the dynamics of \ac{SCI} decoding with consolidations without any constraints on the subspace dimension, which is equivalent to \ac{SCL} decoding with unbounded list size. This relaxation gives more understanding on the complexity vs. performance trade-offs on average. Its practical relevance stems from the complexity-adaptive nature of the decoders for the case of the BEC. Recalling Corollary \ref{cor:list_length_vs_sub_dim}, observe that \ac{SCI} decoding} provides a concrete example of the information-theoretic perspective \ed{since, f}or any $y_1^{N}$, the subspace dimension is $d_{m}(y_1^{N})=H\left(U_{\mathcal{A}^{(m)}}|Y_1^{N}=y_1^{N},U_{\mathcal{F}^{(m)}}\right)$ and \ac{SCI} decoding stores a basis for this subspace instead of listing all possible sequences in it.

Let $\epsilon^{(m)}_{N}\triangleq\Pr\left(p_{\own{m}}(u_{1}^{\own{m}-1})=\nicefrac{1}{2}\right)$, where the implied randomness is due to the received vector. Consider the decoding of information and frozen bits given the observed vector and preceding frozen bits. When an information bit $u_{m}$ is decoded, one of following events occurs:
\begin{itemize}
	\item The information bit is decoded as an erasure and the subspace dimension increases by one, i.e., $d_{m}(y_1^{N})=d_{m-1}(y_1^{N})+1$. Averaged over all $y_1^{N}$, the probability of this event equals $\epsilon_{N}^{(m)}$\cite{CNP20}.
	\item The information bit is decoded as an affine function of the previous information bits and the subspace dimension is unchanged, i.e., $d_{m}(y_1^{N})=d_{m-1}(y_1^{N})$. Averaged over all $y_1^{N}$, the probability of this event equals $1-\epsilon_{N}^{(m)}$\cite{CNP20}.
\end{itemize}

If a frozen $u_m$ is decoded, one of following events occurs:
\begin{itemize}
	\item The decoder returns an erasure for the frozen bit. In this case, revealing the true value of the frozen bit allows decoding to continue, but no new information is provided about preceding information bits. Thus, we have $d_{m}(y_1^{N}) = d_{m-1}(y_1^{N})$. Averaged over all $y_1^{N}$, the probability of this event equals \ed{$\epsilon_{N}^{(m)}$}\cite{CNP20}.
	\item The frozen bit is decoded as an affine function of the previous information bits. Averaged over all $y_1^{N}$, the probability of this event equals $1-\epsilon_{N}^{(m)}$\cite{CNP20}. In this case, revealing the true value of the frozen bit gives a linear equation for a subset of the preceding information bits. If the linear equation is informative, then the subspace dimension decreases by one via a consolidation event, i.e., we have $d_{m}(y_1^{N})=d_{m-1}(y_1^{N})-1$. Otherwise, the dimension is unchanged, i.e., $d_{m}(y_1^{N})=d_{m-1}(y_1^{N})$.
\end{itemize}

At first glance, these rules might appear to tell the whole story. But the erasure rate $\epsilon_{N}^{(m)}$ is averaged over all $y_1^{N}$ whereas predicting the value of $D_{m}$ requires the conditional probability of erasure events given all past observations. 
More importantly, to understand consolidation events, one needs to compute the probability that the obtained equation is informative.

Since we do not have expressions for these quantities,\footnote{Even if we had them exactly, they may be too complicated to be useful.} we use two simplifying approximations \ed{for a given random sequence $D_1,\ldots,D_{m-1}$}. First, we approximate the probability of decoding an erasure for a frozen bit as independent of all past events, i.e., for any $d_1^{m-1}\edd{\triangleq (d_1\left(y_1^N\right),\ldots,d_{m-1}\left(y_1^N\right))}$, we write
\begin{align}
    \Pr\left(\left.p_{m}(u_{1}^{m-1})=\nicefrac{1}{2}\right|D_1^{m-1}=d_1^{m-1}\right) &\approx \epsilon^{(m)}_{N}.
\end{align}
Second, we approximate the probability that an informative equation obtained from consolidation at decoding stage $m$ by $1-2^{-D_{m-1}}$, independent of sequence $D_1,\ldots,D_{m-2}$. This means, for $m\in\mathcal{F}$, we write
\begin{align}
    &\Pr\left(D_{m}=d_{m-1}\!\left|D_1^{m-1} \!=\! d_1^{m-1}, p_{\edd{m}}(u_{1}^{\edd{m}-1})\neq\nicefrac{1}{2}\right.\right) \!\approx\! 2^{-d_{m-1}} \\
    &\approx \!1\!-\!\Pr\left(D_{m}=d_{m-1}\!-\!1\left|D_1^{m-1} = d_1^{m-1}, p_{\edd{m}}(u_{1}^{\edd{m}-1})\neq\nicefrac{1}{2}\right.\right)
\end{align}
\edd{which} comes from modeling the obtained equation and the subset using a uniform random model. Under these assumptions, the random sequence $D_{1},\ldots,D_{N}$ can be modelled by an inhomogeneous Markov chain with transition probabilities $P_{i,j}^{(m)} \ed{\triangleq} \Pr\left(D_{m}=j\,|\:D_{m-1}=i\right)$ where
\begin{equation}\label{eq:Markov_app}
	P_{i,j}^{(m)}\ed{\approx}\begin{cases}
	\epsilon_{N}^{(m)} & \!\!\!\!\!\text{if }m\in\mathcal{\mathcal{A}},~j=i+1\\
	1-\epsilon_{N}^{(m)} & \!\!\!\!\!\text{if }m\in\mathcal{\mathcal{A}},~j=i\\
	\epsilon_{N}^{(m)}+\left(1-\epsilon_{N}^{(m)}\right)2^{-D_{m-1}} & \!\!\!\!\!\text{if }m\in\mathcal{\mathcal{F}},~j=i\\
	\left(1-\epsilon_{N}^{(m)}\right)\left(1-2^{-D_{m-1}}\right) & \!\!\!\!\!\text{if }m\in\mathcal{\mathcal{F}},~j=i-1.
	\end{cases}
\end{equation}

\ed{Consider now decoding of frozen bit $u_m$ based on this Markov chain approximation. We write
\begin{align}
	\bar{D}_m &\approx \E\left[D_{m-1} - \left(1-\epsilon_{N}^{(m)}\right)\left(1-2^{-D_{m-1}}\right)\right]\\
	 &\approx \bar{D}_{m-1} - \left(1-\epsilon_{N}^{(m)}\right)\left(1-2^{-\bar{D}_{m-1}}\right)\label{eq:frozen_rec}
\end{align}
where the last line follows from approximating $\mathbb{E}\left[2^{-D_{m}}\right]$ as $2^{-\bar{D}_m}$. In the case of information bit $u_m$, we have
\begin{align}
	\bar{D}_m &\approx \E\left[\epsilon_{N}^{(m)}\left(D_{m-1}+1\right) + \left(1-\epsilon_{N}^{(m)}\right)D_{m-1}\right]\\
	&= \bar{D}_{m-1}+\epsilon_{N}^{(m)}.\label{eq:info_rec}
\end{align}
By setting $\bar{D}_0\triangleq 0$, \eqref{eq:frozen_rec} and \eqref{eq:info_rec} give the simple recursive approximation
\begin{equation}\label{eq:recursive_app}
	\bar{D}_m\approx\begin{cases}
	\bar{D}_{m-1}+\epsilon_{N}^{(m)} & \!\!\!\!\text{if }m\in\mathcal{A}\\
	\bar{D}_{m-1}-\left(1\!-\!\epsilon_{N}^{(m)}\right)\!\left(\!1-\!2^{-\bar{D}_{m-1}}\right) & \!\!\!\!\text{if }m\in\mathcal{F}.
	\end{cases}
\end{equation}}

%% file: tex/concentration.tex
\section{Concentration for List Decoders}\label{sec:concentration}
This section studies the stochastic convergence properties of the \ac{r.v.} $D_m$. In particular, we show that the required uncertainty\ed{, quantified by $D_m$,} accumulated by SCL \ed{decoding} to keep the correct path on the list concentrates around its mean $\ed{\bar{D}_m}$ for sufficiently large block lengths. 
\subsection{General Approach}
We form a Doob's Martingale by sequentially revealing information about the object of interest (e.g., see \cite{LMSS98,RU01}), which is the conditional entropy in our case. In $N$ consecutive steps, we reveal the random channel realizations. Irrespective of the revealed realization, the change in the \ed{conditional entropy} is bounded by some constant. This lets us use the Azuma-Hoeffding inequality\cite[Thm.~12.6]{Mitzenmacher05} since the channels under consideration are memoryless.
\begin{proposition}
	The sequence of \acp{r.v.} $H_0^{(m)}, H_1^{(m)},\dots,H_N^{(m)}$ where $H_i^{(m)} \triangleq \E\left[D_m|Y_1^i\right]$ is a Doob's Martingale, i.e.,
	\begin{align}
		&H_i^{(m)}\text{is a function of }Y_1^i\label{eq:mart_1}\\
		&\E\left[|D_m|\right] < \infty\label{eq:mart_2}\\
		&H_{i-1}^{(m)} = \E\left[H_{i}^{(m)}|Y_1^{i-1}\right].\label{eq:mart_3}
	\end{align}
\end{proposition}
\begin{proof}
	The statement \eqref{eq:mart_1} follows from the construction of \acp{r.v.} $H_i^{(m)}$ and the definition of conditional expectation. The inequality \eqref{eq:mart_2} follows from the non-negativity of $D_m$ and $\E\left[D_m\right] = H\left(U_{\mathcal{A}^{(m)}}|Y_1^{N},U_{\mathcal{F}^{(m)}}\right)$. Finally, \eqref{eq:mart_3} follows by
	\begin{align}
		\E\left[H_{i}^{(m)}|Y_1^{i-1}\right] &= \E\left[\E\left[D_{m}|Y_1^{i}\right]|Y_1^{i-1}\right] \label{eq:tower1} \\
		&= \E\left[D_{m}|Y_1^{i-1}\right] \label{eq:tower2} \\
		&= H_{i-1}^{(m)} \label{eq:tower3}
	\end{align}
	where \eqref{eq:tower1} and \eqref{eq:tower3} follow from the definition of $H_i^{(m)}$, and \eqref{eq:tower2} from the tower property~\cite[Eq. (C.13)]{Richardson:2008:MCT:1795974}.
\end{proof}
\begin{proposition}\label{prop:lipschitz_bms}
	Consider transmission over a \ed{discrete output} BMS channel satisfying $W(y|x)\geq\delta>0, \forall y\in\mathbb{\mathcal{Y}}$, $\forall x\in\{0,1\}$. Then, for all $i\in[N]$ and all values $y_1^N$ and $\tilde{y}_1^N$ such that $y_{\sim i} = \tilde{y}_{\sim i}$ and $y_i \neq \tilde{y}_i$, the conditional entropy \ed{$d_{m}\left(y_1^N\right) = H\left(U_{\mathcal{A}^{(m)}}|Y_{1}^{N}=y_1^N,U_{\mathcal{F}^{(m)}}\right)$} satisfies
	\begin{equation}
		\left|d_m(y_1^{N})-d_m(\tilde{y}_1^N)\right|\leq 4\left|\log_{2}\delta\right|.
		\label{eq:lipschitz_bms}
	\end{equation}
\end{proposition}
\begin{proof}
	In the following, the \acp{r.v.} are not explicitly written in the probability assignments, e.g., the probabilities are denoted as $p\left(u_1^m,x_1^N,y_1^N\right) = \Pr\left(U_1^m = u_1^m, X_1^N = x_1^N,Y_1^N = y_1^N\right)$. 
	
	The proof starts by writing
	\begin{align}
	&\frac{p\left(u_{\mathcal{A}^{(m)}}|y_1^N,u_{\mathcal{F}^{(m)}}\right)}{p\left(u_{\mathcal{A}^{(m)}}|\tilde{y}_1^N,u_{\mathcal{F}^{(m)}}\right)} \numeq{\text{a}} \frac{p\left(u_1^m,y_1^N\right)}{p\left(y_1^N,u_{\mathcal{F}^{(m)}}\right)}\cdot\frac{p\left(\tilde{y}_1^N,u_{\mathcal{F}^{(m)}}\right)}{p\left(u_1^m,\tilde{y}_1^N\right)}\label{eq:bms_dim1}\\
	&\numeq{\text{b}} \frac{\sum_{x_1^N}p\left(u_1^m,x_1^N,y_1^N\right)}{\sum_{x_1^N}p\left(y_1^N,x_1^N,u_{\mathcal{F}^{(m)}}\right)}\cdot\frac{\sum_{x_1^N}p\left(\tilde{y}_1^N,x_1^N,u_{\mathcal{F}^{(m)}}\right)}{\sum_{x_1^N}p\left(u_1^m,x_1^N,\tilde{y}_1^N\right)}\label{eq:bms_dim2}\\
	&\numeq{\text{c}} \frac{\sum_{x_i}W(y_i|x_i)\sum_{x_{\sim i}}p\left(u_1^m,x_1^N,y_{\sim i}\right)}{\sum_{x_i}W(\tilde{y}_i|x_i)\sum_{x_{\sim i}}p\left(u_1^m,x_1^N,y_{\sim i}\right)} \\ &\qquad\qquad\quad\cdot\frac{\sum_{x_i}W(\tilde{y}_i|x_i)\sum_{x_{\sim i}}p(y_{\sim i}, x_1^N, u_{\mathcal{F^{\text{$(m)$}}}})}{\sum_{x_i}W(y_i|x_i)\sum_{x_{\sim i}}p(y_{\sim i}, x_1^N, u_{\mathcal{F^{\text{$(m)$}}}})}\label{eq:bms_dim4}\\
	&\numeq{\text{d}} \frac{\sum_{x_i}W(y_i|x_i)p(u_1^m,x_i,y_{\sim i})}{\sum_{x_i}W(\tilde{y}_i|x_i)p(u_1^m,x_i,y_{\sim i})} \\ &\qquad\qquad\quad\cdot\frac{\sum_{x_i}W(\tilde{y}_i|x_i)p\left(y_{\sim i},x_i,u_{\mathcal{F}^{(m)}}\right)}{\sum_{x_i}W(y_i|x_i)p\left(y_{\sim i},x_i,u_{\mathcal{F}^{(m)}}\right)}\label{eq:bms_dim5}
	\end{align}
	where ($\text{a}$) follows from Bayes' rule, ($\text{b}$) and ($\text{d}$) from the law of total probability, and ($\text{c}$) from rearranging the summation, Bayes' rule and noting that \ed{$Y_i$ and $(U_1^m,X_{\sim i},Y_{\sim i})$ are independent given $X_i$}. Then, we take the logarithm and absolute value of both sides in \eqref{eq:bms_dim5}. Combining the triangle inequality, i.e., $|a+b|\leq|a|+|b|$, with the fact that each summand is upper bounded by $\max_{y}\left|\log_{2}W(y|0)\right|$, e.g., 
	\begin{equation}
		\left|\log_{2}\sum_{x_i}W(y_i|x_i)p(u_1^m,x_i,y_{\sim i})\right|\leq\max_{y}\left|\log_{2}W(y|0)\right|
	\end{equation}
	we conclude that
	\begin{equation}
	\left|\log_{2}\frac{p\left(u_{\mathcal{A}^{(m)}}|y_1^N,u_{\mathcal{F}^{(m)}}\right)}{p\left(u_{\mathcal{A}^{(m)}}|\tilde{y}_1^N,u_{\mathcal{F}^{(m)}}\right)}\right|\leq 4\max_{y}\left|\log_{2}W(y|0)\right|.\label{eq:bms_dim6}
	\end{equation}
	Since \eqref{eq:bms_dim6} is valid for any $u_1^m$, averaging over all $u_1^m$, combined with the Jensen's inequality, leads to \eqref{eq:lipschitz_bms} by noting $W(y|0)\geq\delta, \forall y\in\mathbb{\mathcal{Y}}$.
\end{proof}

As a result of Proposition \ref{prop:lipschitz_bms}, the following corollary provides a concentration for the logarithm of the list size required to approach the performance of a code under ML decoding when the transmission is over discrete output BMS channels. More precisely, the normalized (with respect to the block length) deviation of the logarithm of the random list size, required to keep the correct codeword in the list, from the average decays exponentially fast.
\begin{corollary}\label{cor:conc_bms}
		For transmission over a \ed{discrete output} BMS channel satisfying $W(y|x)\geq\delta>0, \forall y\in\mathbb{\mathcal{Y}}$, $\forall x\in\{0,1\}$, the \ac{r.v.} $D_m$, $m\in[N]$, concentrates around its mean \ed{$\bar{D}_m$} for sufficiently large block lengths, i.e., for any $\beta>0$, we have
		\begin{equation}\label{eq:concentration_bms}
			\Pr\left\{\frac{1}{N}|D_m-\ed{\bar{D}_m}|>\beta\right\}\leq 2\exp\left(-\frac{\beta^2}{\ed{32}\left|\log_{2}\delta\right|\ed{^2}}N\right).
		\end{equation}
	\end{corollary}
	\begin{proof}
		\ed{Since the channel under consideration is memoryless, $Y_i$, $i\in[N]$, are independent due to the uniform $U_1^N$. Hence, Proposition \ref{prop:lipschitz_bms} implies
		\begin{equation}
		    \left|H_i^{(m)}-H_{i-1}^{(m)}\right| \leq 4\left|\log_{2}\delta\right|,\,\, i\in[N].
		\end{equation}
		The Azuma-Hoeffding inequality\cite[Thm.~\ed{12.4}]{Mitzenmacher05} is applied then by observing that the first element in the martingale is the expectation of $D_m$ and the last one is the \ac{r.v.} itself, i.e., $H_0^{(m)} = \bar{D}_m$ and $H_N^{(m)} = D_m$}.
\end{proof}
Since $\delta=0$ for the BEC, it will be considered separately in the next section.
For the case where $W(y|0)$ is a continuous probability density function on a compact set $\mathcal{Y} \subset \mathbb{R}$, the same proof applies with $\delta = \min_{y \in \mathcal{Y}} W(y|0)$.
But, the proof does not extend to unbounded output alphabets. 

\subsection{The Binary Erasure Channel}
SCI decoding over the BEC is equivalent to solving a system of linear equations with side information depending on the decoding stage. In other words, the decoder has the knowledge of the frozen bits $u_{\mathcal{F}^{(m)}}$ after the decoding stage $m$ as side information.
\begin{proposition}\label{prop:lipschitz}
	For transmission over the BEC, the subspace dimension satisfies the Lipschitz-$1$ condition: for all $i\in[N]$ and all values $y_1^N$ and $\tilde{y}_1^N$ such that $y_{\sim i} = \tilde{y}_{\sim i}$ and $y_i \neq \tilde{y}_i$, the subspace dimension satisfies
	\begin{equation}
		|d_m(y_1^{N})-d_m(\tilde{y}_1^N)|\leq 1.
		\label{eq:lipschitz}
	\end{equation}
\end{proposition}
\begin{proof}
	It suffices to consider the case where $y_i$ is not erased, but $\tilde{y}_i$ is an erasure (i.e., $y_i=x_i$ and $\tilde{y}_i =\,?$). Recall the linear system given as \eqref{eq:linear_system}.
	All compatible vectors $(u_{\mathcal{A}^{(m)}},u_{m+1}^N)$, $m\in[N]$, with \eqref{eq:linear_system} form an affine subspace. The dimension of this subspace is equal to $d^\prime_N(y_1^N) = \ed{N}-|\mathcal{F}^{(m)}|-\rank(\G^\prime_{N})$. Since removing one more column of $\G^\prime_{N}$ (and also of $\G^\prime_{\mathcal{F}^{(m)}}$) cannot decrease the rank by more than one, we have
	\begin{equation}\label{eq:dim_bound}
		d^\prime_N(y_1^N)\leq d^\prime_N(\tilde{y}_1^N)\leq d^\prime_N(y_1^N)+1.
	\end{equation}
	Hence, the number of compatible vectors $(u_{\mathcal{A}^{(m)}},u_{m+1}^N)$ with \eqref{eq:linear_system} is (at most) doubled or unchanged.
	
	We are interested in the subspace dimension $d_m(y_1^N)$. This is equal to the number of different subvectors $u_{\mathcal{A}^{(m)}}$ of all compatible $(u_{\mathcal{A}^{(m)}},u_{m+1}^N)$ with \eqref{eq:linear_system}. Using \eqref{eq:dim_bound}, one concludes that the number of different vectors $u_{\mathcal{A}^{(m)}}$ either increases by a factor of $2$ or does not change, resulting in \eqref{eq:lipschitz}.
\end{proof}
\begin{corollary} \label{cor:conc_bec}
		The subspace dimension $D_m$ concentrates around its mean \ed{$\bar{D}_m$} for sufficiently large block lengths, i.e., for any $\beta>0$, we have
		\begin{equation} \label{eq:conc_bec}
		\Pr\left\{\frac{1}{N}|D_m-\ed{\bar{D}_m}|>\beta\right\}\leq 2\exp\left(-\frac{\beta^2}{2}N\right).
		\end{equation}
	\end{corollary}
	\begin{proof}
		As for Corollary \ref{cor:conc_bms}, apply the Azuma-Hoeffding inequality\cite[Thm.~\ed{12.4}]{Mitzenmacher05} via Proposition \ref{prop:lipschitz}.
\end{proof}

\begin{remark}
Let $\rho = \lceil \log_2 m \rceil$ and $N_0 = 2^\rho$.
Due to the recursive structure of SCL decod\ed{ing}, the statistics of $D_m$ are the same for all $N\geq N_0$ if the first $N_0$ frozen bits are the same.
Thus, Corollary~\ref{cor:conc_bec} remains valid if we replace~\eqref{eq:conc_bec} by
		\begin{equation}
		\Pr\left\{\frac{1}{N_0}|D_m-\ed{\bar{D}_m}|>\beta\right\}\leq 2\exp\left(-\frac{\beta^2}{2}N_0\right).
		\end{equation}
This provides a significant improvement when $N_0 \ll N$.
The same idea can also be applied to Corollary~\ref{cor:conc_bms} but the value of $\delta$ must be modified as well.
\end{remark}

\edd{Note that the bounds of the form \eqref{eq:concentration_bms} are typically loose (see \cite[Sec. IV]{RU01} for a discussion on tightness of the concentration results for the performance of a randomly chosen \ac{LDPC} code around the ensemble average). Nevertheless, such analysis shows that the mean $\bar{D}_m$ under consideration is meaningful.}

%% file: tex/numerical_results.tex
\section{Simulation Results}
\label{sec:numerical}
This section provides simulation results for some constructions with dynamic frozen bits. In particular, we consider instances from an ensemble of modified RM codes, namely \own{\ac{dRM}} codes\cite{CNP20}.
\begin{definition}\label{def:dRM}
	The \own{\ac{dRM}}$(r,m)$ ensemble, denoted by $\ensemble(r,m)$, is the set of all codes, specified by set $\mathcal{A}$ of the RM$(r,m)$ code and setting, for $i\in\mathcal{F}$,
	\begin{equation}
	    u_i = \begin{cases}
	        \sum_{j\in\mathcal{A}^{(i-1)}} v_{j,i}u_j & \text{if }\mathcal{A}^{(i-1)}\neq\emptyset\\
	        0 & \text{otherwise}.
	    \end{cases}    
	\end{equation}
	with all possible $v_{j,i}\in\{0,1\}$ and $\mathcal{A}^{(0)}\triangleq\emptyset$.
\end{definition}

Recently, Ar{\i}kan introduced \ac{PAC} codes\cite{arikan2019sequential}, which can be represented as a polar code with dynamic frozen bits \cite{YFV21:PAC}. The rate-profiling choice of a PAC code is directly reflected in the frozen index set of its polar code representation\cite{YFV21:PAC}. Thus, if $\mathcal{A}$ of an RM$(r,m)$ code is chosen as the rate-profiling (as in \cite{arikan2019sequential}), then the corresponding PAC code becomes an instance from $\ensemble(r,m)$. Another instance is the RM$(r,m)$ code. 

\subsection{The Binary-Input Additive White Gaussian Channel}\label{sec:numerical_awgn}
Fig.~\ref{fig:list_evo0} shows simulation results for a random instance from $\ensemble(3,7)$ and a novel design (based on suggestions in Remark \ref{remark:design}) under SCL decoding with $L=2^{14}$ and $E_b/N_0=0.5$ together with the upper and lower bounds \eqref{eq:bounds_UB} and \eqref{eq:bounds_LB} on $\ed{\bar{D}_m}$. The proposed code takes the set $\mathcal{A}_{\mathrm{RM}}$ of the $(128,64)$ RM code and obtains a new set as $\mathcal{A} = (\mathcal{A}_{\mathrm{RM}}\setminus\{30,40\})\cup\{1,57\}$, i.e., $u_{\{30,40\}}$ are frozen and $u_{\{1,57\}}$ are unfrozen, where each frozen bit is still set to a random linear combination of preceding information bit(s). \ed{Fig.~\ref{fig:list_evo0} validates the bounds \eqref{eq:bounds_LB}, \eqref{eq:bounds_UB}, \eqref{eq:bound2} and \eqref{eq:bound3}. Note that we set the parameter $\alpha=0.94$ in \eqref{eq:bound2} to provide a robust estimate by capturing the near misses, which happen if there are decoding paths with probabilities slightly larger than that of the correct path. To understand this better, consider the proposed design where $u_1$ is an information bit. If one sets $\alpha = 1$, we get
\begin{equation}
    \E\left[\log_{2}|\mathcal{S}_1^{(1)}|\right] \approx 0.5\log_21+0.5\log_22 = 0.5\label{eq:expected_log_rank_first_bit}
\end{equation}
which follows from $\bar{D}_m \approx 1$. Observing Fig. \ref{fig:list_evo0}, we obtain $\E\left[\log_{2}|\mathcal{S}_{0.94}^{(1)}|\right]\approx 1$. Therefore, tightness of \eqref{eq:bound2}, especially at early decoding stages, is impacted by the choice of $\alpha$.\footnote{\ed{One may further reduce the threshold $\alpha$ for inclusion to find a better match of $\E\left[\log_{2}|\mathcal{S}_\alpha^{(m)}|\right]$ to $\bar{D}_m$ for the entire range.}} Our numerical results show that the curve for $\E\left[\log_{2}|\mathcal{S}_\alpha^{(m)}|\right]$ is more robust to changes in $\alpha$ at late decoding stages, i.e., for larger values of $m$. This means that the near misses happen at early decoding stages more often. In addition, observe that} \eqref{eq:bounds_LB} closely tracks the simulation for $m\leq 50$ \ed{and it is easy to compute via \edd{standard methods, e.g., we used Gaussian approximation of density evolution\cite{Trifonov:2012}}, which further motivate using it for code design. Following Remark~\ref{remark:design}, we chose a relatively small $\nicefrac{E_b}{N_0}$ for the analysis, e.g., close to the Shannon limit ($\sim0.189$ dB) for rate-$\nicefrac{1}{2}$ codes, since we are after the rare-event probabilities.}
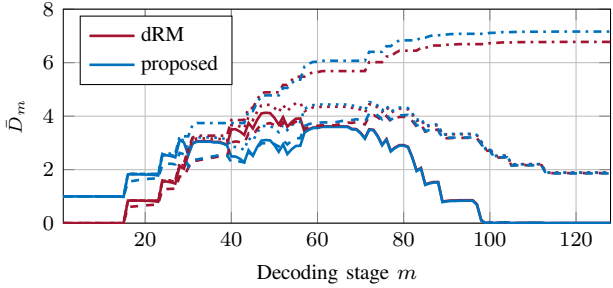
\begin{figure}    
	\centering
	\input{figures/rank_evo_polar.tex}
	\caption{$\ed{\bar{D}_m}$ vs. $m$ at $E_b/N_0=0.5$ dB for $(128,64)$ codes (dash-dotted: upper bound \eqref{eq:bounds_UB}, solid: lower bound \eqref{eq:bounds_LB}, dotted: $\ed{\bar{D}_m}$ via simulation, dashed: $\E[\log_{2}|\mathcal{S}_{0.94}^{(m)}|]$ via simulation).}\label{fig:list_evo0}
\end{figure}

\ed{The changes leading to the new design help} especially for the considered list size $L=32$ \ed{by inspecting Fig. \ref{fig:performance}. The proposed code outperforms the \own{\ac{dRM}} code especially at higher SNR values with this list size.} The reason is illustrated by the lower bounds on \ed{$\ed{\bar{D}_m}$} in Fig.~\ref{fig:list_evo0}. In addition to having a smaller peak value, this peak occurs for the proposed design later than for the \own{\ac{dRM}} code. This helps for the proposed code to \ed{not loose the correct path at early decoding stages, and hence,} to keep \ed{it} in the list towards the end for small list sizes, e.g., $L=32$. If the list size is further decreased, then having $u_1$ as information bit can cause a degradation.\footnote{\ed{In particular, if $L=1$, \eqref{eq:expected_log_rank_first_bit} suggests that the correct path would be lost roughly half of the time already after first decoding stage.}} This validates the analysis illustrated in Fig.~\ref{fig:list_evo0}. The performance for the 5G design \ed{employing the \ac{CRC}-$11$ defined by the generator polynomial $g(x) = x^{11} + x^{10} + x^{9} + x^5 + 1$~\cite[Sec. 5.1]{5G20},\cite{BCL21}} under SCL decoding with $L=32$ is $0.4$ dB worse than the proposed design at a \own{\ac{BLER}} around $10^{-4}$. When SCL \ed{decoding} with $L=128$ is considered, both codes \ed{outperform the 5G design by no less than $0.25$ dB at all \acp{BLER} considered. In particular, they perform} within $0.15$ dB of the \ac{RCU} bound~\cite[Thm.~16]{Polyanskiy10:BOUNDS} at \own{the \ac{BLER}} of $10^{-5}$ and they almost match the simulation-based ML lower bounds\cite{tal15}, denoted as ML LB in the figure. \ed{Note that the \ac{PAC} code perform very close to the \ac{dRM} code under \ac{SCL} decoding with $L\in\{32,128\}$~\cite[Fig. 1]{YFV21:PAC}.} The \ac{MC} bound \cite[Thm.~28]{Polyanskiy10:BOUNDS} is also provided.
\begin{figure}
	\centering
	\input{figures/RMsubcode.tex}
	\vspace*{-5mm}
	\caption{BLER vs. SNR for $(128,64)$ codes.}\label{fig:performance}
\end{figure}
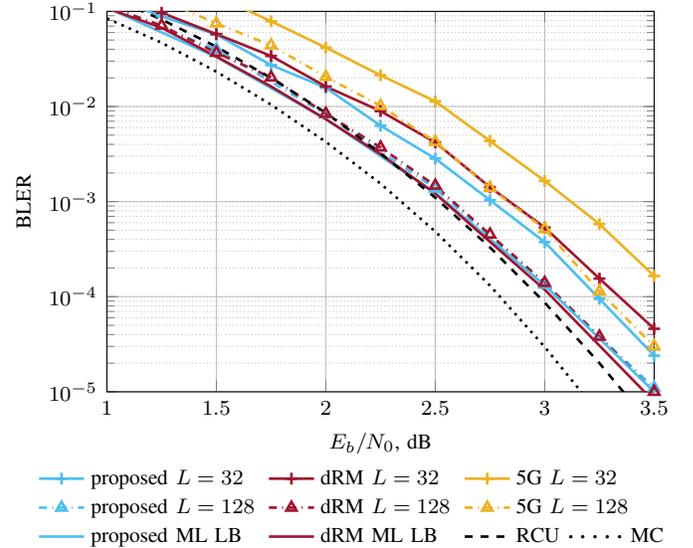

\begin{figure}
	\centering
	\input{figures/rank_evo_polar2.tex}
	\caption{Lower bound \eqref{eq:bounds_LB} on $\ed{\bar{D}_m}$ \own{vs. $m$} at $E_b/N_0=0.5$ dB for $(512,256)$ codes.}\label{fig:list_evo512_256}
\end{figure}
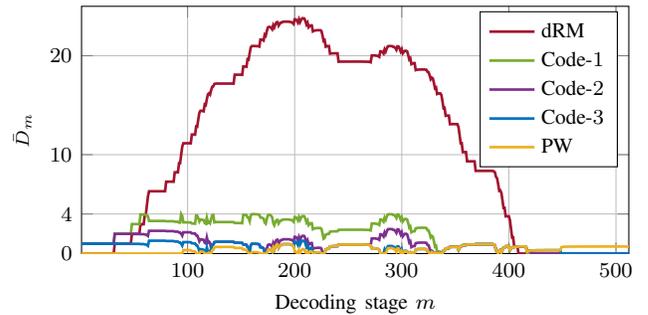
Next, consider moderate-length codes, e.g., $(512,256)$ codes, which are more challenging to design if the decoders are restricted to be of low- to moderate-complexity, i.e., $L\leq 1024$\cite[Sec. 5.2]{Coskun18:Survey}\cite{ILX:survey16}. Fig.~\ref{fig:list_evo512_256} provides the bounds \eqref{eq:bounds_LB} for \own{\ac{dRM}} codes and three novel designs. The peak of the lower bound corresponding to the \own{\ac{dRM}} code gets close to $25$ and recall that this quantity is related to the logarithm of the required list size \ed{on average}. This explains why SCL \ed{decoding} needs very large list sizes for a good performance when used for the RM$(4,9)$ (or a \own{\ac{dRM}}$(4,9)$) code\cite{Mondelli14}. At the other extreme, the lower bound is provided for the construction based on the \ac{PW} method with $\beta = 2^{\nicefrac{1}{4}}$\cite{beta_polar}, which is more suitable for SCL decoding with small list sizes. The idea behind the designs is similar to the length-$128$ case: we start from the information positions of an RM code, modify the positions to lower the peak value and keep the curve flat so that there is enough entropy kept on the list to make use of reliable frozen positions for a good performance. To this end, we also introduced $u_1$ as information bit in all three designs. This would harm the performance if the list size is very small, e.g., $L\leq 4$. In modifying the designs, we used the information positions from the \ac{PW} construction. The information positions for the designs are provided in the appendix.

\begin{figure}
	\centering
	\input{figures/proposed_512_256.tex}
	\vspace*{-5mm}
	\caption{BLER vs. SNR for $(512,256)$ codes.}\label{fig:performance2}
\end{figure}
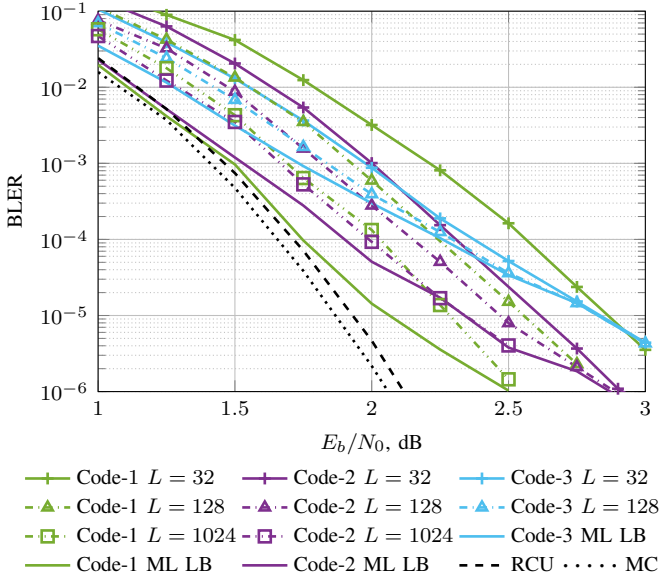
\begin{figure}
	\centering
	\input{figures/proposed_vs_5G_512_256_CRC7.tex}
	\vspace*{-5mm}
	\caption{\edd{BLER vs. SNR for $(512,256)$ concatenated polar codes with CRC-$7$ compared to Code-$3$.}}\label{fig:5G_CRC7}
\end{figure}
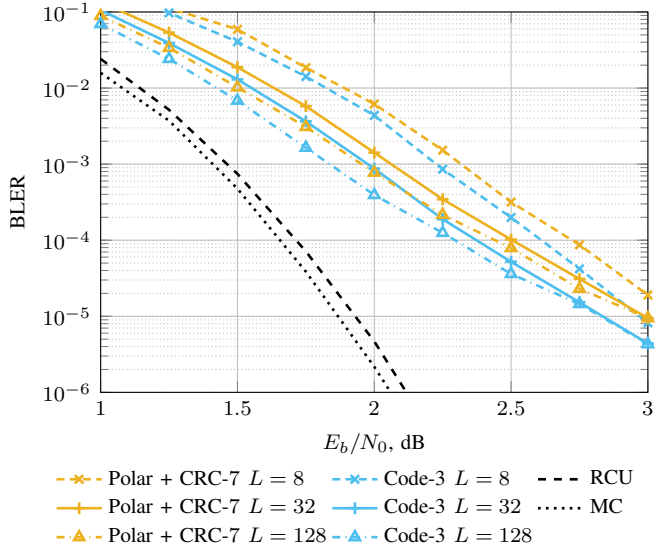
\begin{figure}
	\centering
	\input{figures/proposed_vs_5G_512_256.tex}
	\vspace*{-5mm}
	\caption{BLER vs. SNR for $(512,256)$ 5G codes compared to \own{Code-$2$}.}\label{fig:performance3}
\end{figure}
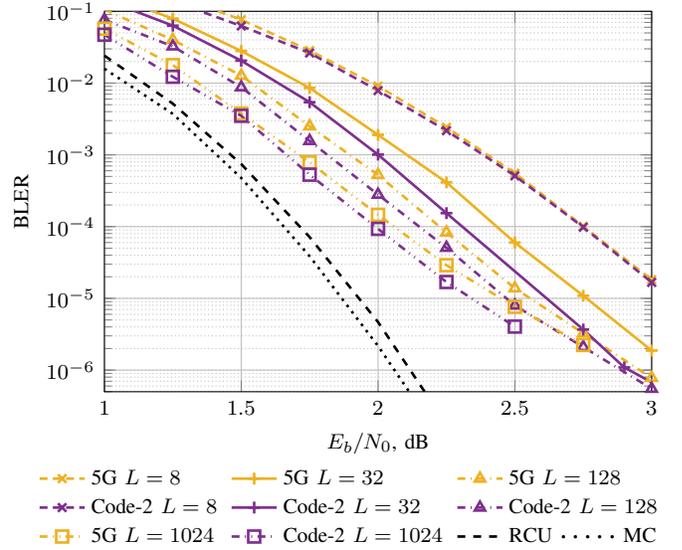
\begin{figure}
	\centering
	\input{figures/proposed_vs_5G_512_256_CRC16.tex}
	\vspace*{-5mm}
	\caption{\edd{BLER vs. SNR for $(512,256)$ concatenated polar codes with CRC-$16$ compared to Code-$3$.}}\label{fig:5G_CRC16}
\end{figure}
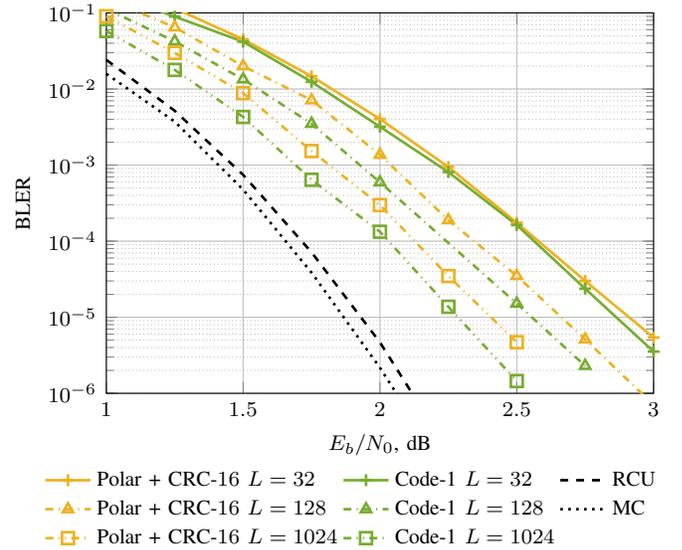
Fig.~\ref{fig:performance2} compares the performance of three designs under different list sizes. Code-$1$ requires the largest list size to get closer to its ML performance. When a large list size is adopted, e.g., $L\in\{1024\}$, it performs within $0.4$ dB of the \ac{RCU} bound at \own{\acp{BLER}} around $10^{-6}$, outperforming the non-binary LDPC code defined over $\mathbb{F}_{256}$ which has a higher decoding complexity\cite{ILX:survey16}. Nevertheless, even with $L=1024$, there is a non-negligible gap to the ML lower bound at \own{\acp{BLER}} above $10^{-6}$. Code-$2$ is competitive for a wide range of list sizes, i.e., $L\in\edd{[8,1024]}$. In particular, it performs within $0.75$ dB from the RCU bound down to the \own{\ac{BLER}} of $10^{-6}$ under SCL decoding with $L=32$, outperforming the 5G design \ed{employing the \ac{CRC}-$11$ with the generator polynomial $g(x) = x^{11} + x^{10} + x^{9} + x^5 + 1$} by around $0.2$ dB under the same list size \ed{at a \ac{BLER} close to $10^{-6}$} (see Fig.~\ref{fig:performance3}). \ed{For $E_b/N_0\geq 2.75$, Code-$2$ under \ac{SCL} decoding with $L=32$ performs as good as the 5G design under \ac{SCL} decoding with $L=128$.} When an even smaller list size considered, e.g., $L=8$, then Code-$2$ and the 5G design are indistinguishable while Code-$3$, which is expected to \ed{require smaller} list sizes \ed{to approach its \ac{ML} performance} (see Fig.~\ref{fig:list_evo512_256}), outperforms them at all \own{\acp{BLER}} \edd{as can be seen in Fig.~\ref{fig:5G_CRC7}, where it is also compared to a polar code concatenated with the \ac{CRC}-$7$ with the generator polynomial $g_7(x) = x^7 + x^6 + x^5 + x^2 + 1$.\footnote{\edd{The polynomial is taken from \cite{Yuan19} as it provides the best performance for the $(128,64)$ case although it may not be optimal for the $(512,256)$ code.}}} With a relatively small list size, e.g., $L=32$, Code-$3$ reaches to its ML performance at \own{the \ac{BLER}} of $10^{-5}$ or less. \edd{Fig.~\ref{fig:5G_CRC16} compares Code-$1$ to a polar code concatenated with the CRC-$16$ specified in \cite[Sec. 5.1]{5G20} with the generating polynomial $g_{16}(x) = x^{16}+x^{12}+x^{5}+1$. For the considered list sizes, Code-$1$ outperforms the modified polar code. In addition,} Code-$1$ performs very similar to the 5G design when $L=128$, it provides sizeable gains, e.g., $0.35$ dB, at \acp{BLER} close to $10^{-6}$ when $L=1024$. \edd{Note finally that for all polar codes (irrespective of the chosen \ac{CRC} length) provided as reference in this work the indices of frozen bits are selected according to the 5G standard~\cite[Sec. 5.1]{5G20},\cite{BCL21}.}

\subsection{The Binary Erasure Channel}
\label{sec:num_results_bec}
To understand the accuracy of the analysis and approximations presented in Sections~\ref{sec:analysis} and \ref{sec:concentration}, we simulated SCI \ed{decoding} with consolidations \ed{for a \emph{\own{\ac{dRM}} ensemble sequence}, where the $\ell$-th ensemble in the sequence is defined to be $\ensemble(\ell,2\ell+1)$. Observe that each ensemble in the sequence consists of rate-$\nicefrac{1}{2}$ codes. Note that as the block length changes, a random instance from the corresponding ensemble is picked for the simulations and the numerical results are quite similar for different instances chosen randomly. These simulations provide realizations of the random process $D_{1},\ldots,D_{N}$ and we compare their mean to the theoretical predictions \eqref{eq:Markov_app} and \eqref{eq:recursive_app} in \edd{Figs.} \ref{fig:inac_evo_512} and \ref{fig:inac_evo_8192} for the case of $N=512$ and $N=8192$, respectively. These results show that for a random code in $\ensemble(\ell,2\ell+1)$ the simulation mean is close to the simple \edd{recursive formula given as \eqref{eq:recursive_app}. Observe also that the mean computed by approximating the probabilities via Markov model} \eqref{eq:Markov_app} matches the mean slightly better for large values of $m$. With an increasing blocklength, \eqref{eq:Markov_app} and \eqref{eq:recursive_app} match the simulations better for large values of $m$.}

\begin{figure}    
	\centering
	\input{figures/inac_evo_512.tex}
	\caption{\ed{$\Bar{D}_m$ vs. $m$ at $\epsilon = 0.48$ for an instance from $\ensemble(4,9)$.}}\label{fig:inac_evo_512}
\end{figure}
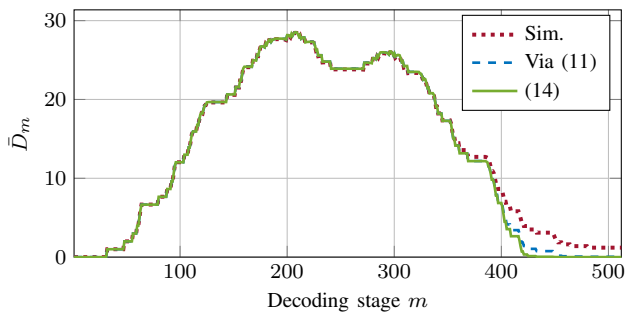
\begin{figure}[t]   
	\centering
	\input{figures/inac_evo_8192.tex}
	\caption{\ed{$\Bar{D}_m$ vs. $m$ at $\epsilon = 0.48$ for an instance from $\ensemble(6,13)$.}}\label{fig:inac_evo_8192}
\end{figure}
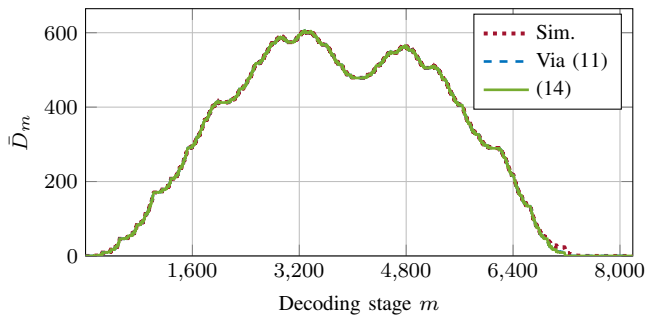
One weakness of these bounds is that the channel variation (e.g., in the number of erasures) significantly increases the variation in $D_{1}^{N}$\ed{, especially for small to medium blocklengths, e.g., recall Fig.~\ref{fig:inac_evo_512} for the case of $N=512$}. To highlight the similarity between the theory and simulation \ed{even for such cases}, we use a fixed-weight BEC that chooses a random pattern with exactly $\texttt{round}(N\epsilon)$ erasures. To motivate this, note that density evolution naturally captures the typical behavior of the analyzed system~\cite{Richardson:2008:MCT:1795974}. Fig.~\ref{fig:inac_evo} shows that\ed{, even for the $(512,256)$ \ac{dRM} code, the simulation mean is close to the analysis for the entire range of $m$}. The $15$ random simulation traces lie largely within the $90\%$ confidence range of the Markov chain analysis \own{\eqref{eq:Markov_app}}.
\begin{figure}[t]
	\centering
	\input{figures/inac_evo.tex}
	\vspace*{-5mm}
	\caption{$\ed{\bar{D}_m}$ vs. $m$ for randomly permuted $246$ erasures.}\label{fig:inac_evo}
\end{figure}
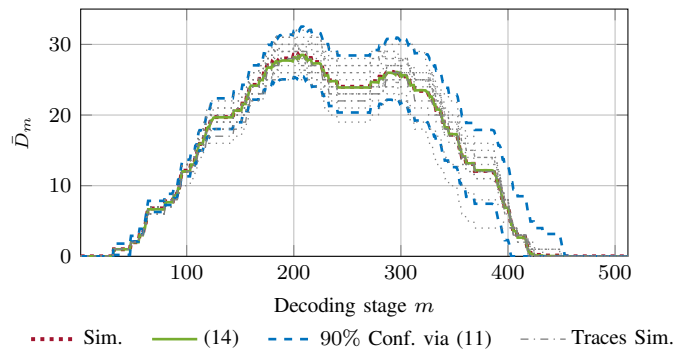

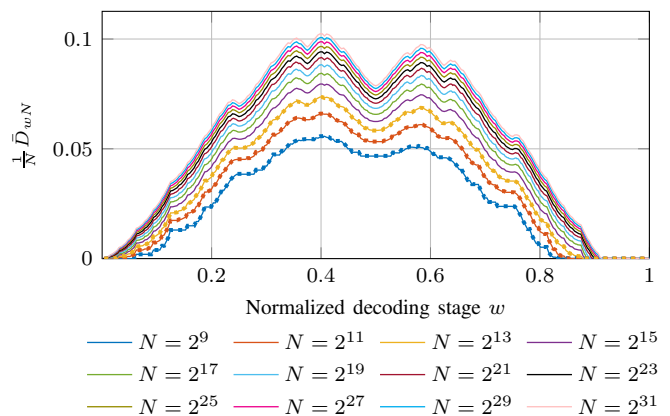
\begin{figure}    
	\centering
	\input{figures/list_growth_rate.tex}
	\vspace*{-5mm}
	\caption{$\frac{1}{N}\ed{\bar{D}}_{wN}$ vs. $w$ (solid: using \eqref{eq:recursive_app} with an erasure probability $\epsilon=0.48$, dashed: simulations for the fixed-weight BEC with $\texttt{round}(N\epsilon)$ erasures) for instances from \own{\ac{dRM}} code ensembles $\ensemble(\ell,2\ell+1)$, $\ell\in\{4,5,\dots,15\}$.} \label{fig:list_growth}
\end{figure}
Next, we study how the subspace dimension behaves as the block length increases \ed{for the introduced \ac{dRM} ensemble sequence}. Let $w\triangleq\frac{m}{N}$, $m\in[N]$, be the normalized decoding stage. Fig.~\ref{fig:list_growth} provides the normalized dimension $\frac{1}{N}\ed{\bar{D}}_{wN}$ as a function of $w$ for the samples of \own{\ac{dRM}} ensemble sequence with different block lengths, from $N=2^9$ up to $N=2^{31}$. \ed{We stress again the accuracy of the analysis \eqref{eq:recursive_app} validated by simulations} up to $N=2^{13}$ and we believe that the results for larger block lengths are also accurate. The asymptotic behavior of $\frac{1}{N}d_{wN}$ gives the asymptotic decoding complexity of an ML decoder implemented via SCI \ed{decoding}. This provides insight into the asymptotic decoding complexity of RM codes to achieve the capacity over the BEC\cite{kudekar17} for two reasons: first, RM code is a member of the ensemble, and second, the simulation results look very similar for RM codes as well up to $N=2^{13}$. Other decoding algorithms might give improvements on the complexity, but we are not aware of a lower-complexity ML decoder than SCI decoding for RM codes \ed{in the case of the BEC}. Fig.~\ref{fig:list_growth} shows that the convergence seems to be rather slow for the defined sequence. Interesting directions include understanding what happens to $\frac{1}{N}\ed{\bar{D}}_{wN}$ as $N\rightarrow\infty$ analytically and trying to define code sequences where $\max_{w}\frac{1}{N}\ed{\bar{D}}_{wN}$ is significantly \ed{smaller} than that of \own{\ac{dRM}} codes, but still perform competitively.

\ed{Next, Fig.~\ref{fig:dm_dist} provides the \ac{p.m.f.} for $\frac{1}{N}D_{\left[0.4N\right]}$, where $[wN]$ is the nearest integer to $wN$, $w\in(0,1]$. The parameter $w$ is set to $0.4$ since the mean analysis given in Fig.~\ref{fig:list_growth} shows that, for the considered codes, the mean reaches to its maximum around $w=0.4$. Interestingly, the \acp{p.m.f.} concentrate around the mean $\frac{1}{N}\Bar{D}_{\left[0.4N\right]}$ with increasing block length.}
\begin{figure}[t]
	\centering
	\input{figures/dm_distribution.tex}
	\caption{\ed{$\Pr\left(\frac{1}{N}D_{\left[0.4N\right]} = d\right)$ vs. $d$ for an erasure probability $\epsilon=0.48$ for instances from \own{\ac{dRM}} code ensembles $\ensemble(\ell,2\ell+1)$, $\ell\in\{4,5,6\}$.}} \label{fig:dm_dist}
\end{figure}
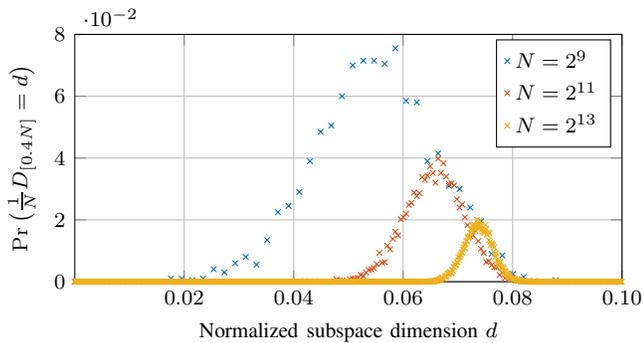

Finally, Fig.~\ref{fig:performance_dynamic} compares the ML decoding performance of Code-$1$ designed in Section~\ref{sec:numerical_awgn} to that of the RM$(4,9)$ and a random instance from $\ensemble(4,9)$ \ed{under} SCI \ed{decoding}.\footnote{\ed{Note that the performance for the $(512,256)$ RM and dRM codes were also provided in \cite{CNP20}. Here, we provide the performance curve for an exemplary construction in addition to Fig. \ref{fig:inac_evo_RM_vs_design}, which illustrate the competitive performance of Code-$1$ for wide range of erasure probabilities while requiring much smaller subspace dimensions under SCI decoding with consolidations (or list lengths under SCL decoding) to achieve ML decoding.}} We declare an error whenever the linear system corresponding to the channel output does not provide a unique solution. The code performs as well as the RM$(4,9)$ code for frame error rates equal to $10^{-5}$ or higher with much lower average subspace dimension during decoding as illustrated in Fig.~\ref{fig:inac_evo_RM_vs_design}. Thus, we can match random code performance (like \own{\ac{dRM}}) down to the \own{\ac{BLER}} of $10^{-4}$ while performing well with smaller number of inactivations (i.e., shorter list sizes). For \own{\ac{BLER}} of $10^{-5}$ or less, it experiences more severe error floor compared to the RM$(4,9)$ code. The Singleton bound~\cite{Singleton64} and the Berlekamp's random coding \own{(BRC)} bound~\cite{Berlekamp80} are provided as reference.
\begin{figure}
	\centering
	\input{figures/RMsubcode_bec.tex}
	\vspace*{-5mm}
	\caption{BLER vs. $\epsilon$ for $(512,256)$ codes.}\label{fig:performance_dynamic}
\end{figure}
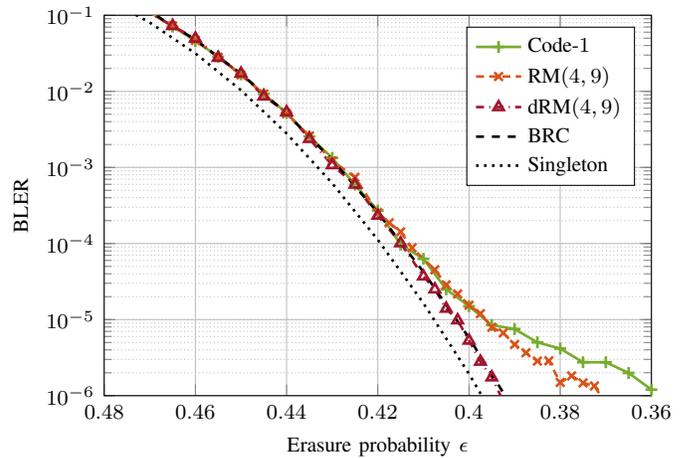
\begin{figure}[t]   
	\centering
	\input{figures/inac_evo_RM_vs_design.tex}
	\caption{$\ed{\bar{D}_m}$ vs. $m$ for Code-$1$ compared to the RM code (solid: using \eqref{eq:recursive_app} with an erasure probability of $\epsilon=0.48$, dashed: simulations for the fixed-weight BEC with $246$ erasures).} \label{fig:inac_evo_RM_vs_design}
\end{figure}
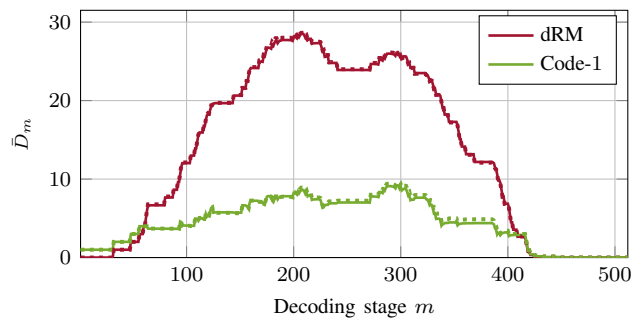

%% file: figures/rank_evo_polar.tex
\begin{tikzpicture}
\begin{axis}[
mark options={solid,scale=0.75},
width=1*\linewidth,
height=0.5*\columnwidth,
title style={font=\footnotesize,align=center},
legend cell align=left,
legend style={font=\footnotesize},
legend columns=1,
legend pos=north west,
ylabel near ticks,
xlabel near ticks,
xmin=1,
xmax=128,
ymin=0,
ymax=8,
xlabel={\textcolor{black}{\own{Decoding stage} $m$}},
ylabel={\textcolor{black}{\scriptsize{$\ed{\bar{D}_m}$}}},
grid=both,
label style={font=\footnotesize},
tick label style={font=\footnotesize},
]

\addplot[color=myParula07Red,line width = 1pt,solid] table[x=index,y=entropy] {figures/results_LB_df_RM_n=128_SNR=0.5.txt};
\addlegendentry{\own{\ac{dRM}}}

\addplot[color=myParula01Blue,line width = 1pt,solid] table[x=index,y=entropy] {figures/results_LB_henry_propsed_df_RM_n=128_SNR=0.5_L=16192.txt};
\addlegendentry{proposed}

\addplot[color=myParula07Red,line width = 1pt,dotted] table[x=index,y=entropy] {figures/results_H_num_df_RM_n=128_SNR=0.5_L=16384.txt};

\addplot[color=myParula01Blue,line width = 1pt,dotted] table[x=index,y=entropy] {figures/results_H_num_henry_propsed_df_RM_n=128_SNR=0.5_L=16384.txt};

\addplot[color=myParula07Red,line width = 1pt,dashed] table[x=index,y=entropy] {figures/results_ElogRank_df_RM_n=128_SNR=0.5_L=16384.txt};

\addplot[color=myParula01Blue,line width = 1pt,dashed] table[x=index,y=entropy] {figures/results_ElogRank_henry_propsed_df_RM_n=128_SNR=0.5_L=16384.txt};

\addplot[color=myParula07Red,line width = 1pt,dash dot] table[x=index,y=entropy] {figures/results_UB_df_RM_n=128_SNR=0.5.txt};

\addplot[color=myParula01Blue,line width = 1pt,dash dot] table[x=index,y=entropy] {figures/results_UB_henry_propsed_df_RM_n=128_SNR=0.5_L=16192.txt};

\end{axis}
\end{tikzpicture}

%% file: figures/RMsubcode.tex
\begin{tikzpicture}
\begin{semilogyaxis}[
mark options={solid,scale=0.75},
width=1*\linewidth,
height=0.75*\columnwidth,
title style={font=\footnotesize,align=center},
legend cell align=left,
legend style={font=\footnotesize},
legend columns=3,
legend style={at={(0.45,-0.175)},anchor=north,draw=none,/tikz/every even column/.append style={column sep=1mm},cells={align=left}},
ylabel near ticks,
xlabel near ticks,
xmin=1,
xmax=3.5,
ymin=1e-5,
ymax=1e-1,
xlabel={\textcolor{black}{$E_b/N_0$, dB}},
ylabel={\textcolor{black}{BLER}},
grid=both,
major grid style={solid,draw=gray!50},
minor grid style={densely dotted,draw=gray!50},
label style={font=\footnotesize},
tick label style={font=\footnotesize},
]

\addplot[myparula61] table[x=ebn0,y=P] {figures/128_64_df_RMproposed_henry_awgn_L=32.txt};
\addlegendentry{proposed $L=32$}
\addplot[myparula71] table[x=ebn0,y=P] {figures/128_64_df_RM_awgn_L=32.txt};
\addlegendentry{\own{\ac{dRM}} $L=32$}
\addplot[myparula31] table[x=ebn0,y=P] {figures/128_64_5G_awgn_L=32_CRC11.txt};
\addlegendentry{5G $L=32$}

\addplot[myparula64] table[x=ebn0,y=P] {figures/128_64_df_RMproposed_henry_awgn_L=128.txt};
\addlegendentry{proposed $L=128$}
\addplot[myparula74] table[x=ebn0,y=P] {figures/128_64_df_RM_awgn_L=128.txt};
\addlegendentry{\own{\ac{dRM}} $L=128$}
\addplot[myparula34] table[x=ebn0,y=P] {figures/128_64_5G_awgn_L=128_CRC11.txt};
\addlegendentry{5G $L=128$}

\addplot[color=myParula06LightBlue,line width=1pt,solid] table[x=ebn0,y=P] {figures/128_64_df_RMproposed_henry_ML_LB_awgn.txt};
\addlegendentry{proposed ML LB}
\addplot[color=myParula07Red,line width=1pt,solid] table[x=ebn0,y=P] {figures/128_64_df_RM_ML_LB_awgn.txt};
\addlegendentry{\own{\ac{dRM}} ML LB}

\addplot[color=black,line width=1pt,dashed] table[x=snr,y=P] {figures/128_64_awgn_rcu.txt};
\addlegendentry{RCU \textcolor{black}{$\boldsymbol{\cdot\cdot\cdot\cdot\cdot}$} MC}
\addplot[color=black,line width=1pt,dotted] table[x=snr,y=P] {figures/128_64_awgn_mc.txt};

\end{semilogyaxis}
\end{tikzpicture}

%% file: figures/rank_evo_polar2.tex
\begin{tikzpicture}
\begin{axis}[
mark options={solid,scale=0.75},
width=1*\linewidth,
height=0.55*\columnwidth,
title style={font=\footnotesize,align=center},
legend cell align=left,
legend style={font=\footnotesize},
legend columns=1,
ylabel near ticks,
xlabel near ticks,
xmin=1,
xmax=512,
ymin=0,
ymax=25,
xlabel={\textcolor{black}{\own{Decoding stage} $m$}},
ylabel={\textcolor{black}{\scriptsize{$\ed{\bar{D}_m}$}}},
ytick={0, 4, 10,  20,  30},
grid=both,
label style={font=\footnotesize},
tick label style={font=\footnotesize},
]

\addplot[color=myParula07Red,line width = 1pt,solid] table[x=index,y=entropy] {figures/results_LB_df_RM_n=512_SNR=0.5.txt};
\addlegendentry{\own{\ac{dRM}}}

\addplot[color=myParula05Green,line width = 1pt,solid] table[x=index,y=entropy] {figures/results_LB_hybrid_n=512_SNR=0.5.txt};
\addlegendentry{Code-$1$}

\addplot[color=myParula04Purple,line width = 1pt,solid] table[x=index,y=entropy] {figures/results_LB_hybrid_small_list_n=512_SNR=0.5.txt};
\addlegendentry{Code-$2$}

\addplot[color=myParula01Blue,line width = 1pt,solid] table[x=index,y=entropy] {figures/results_LB_proposed_n=512_SNR=0.5.txt};
\addlegendentry{Code-$3$}

\addplot[color=myParula03Yellow,line width = 1pt,solid] table[x=index,y=entropy] {figures/results_LB_df_beta_n=512_SNR=0.5.txt};
\addlegendentry{PW}

\end{axis}
\end{tikzpicture}

%% file: figures/proposed_512_256.tex
\begin{tikzpicture}
\begin{semilogyaxis}[
mark options={solid,scale=0.75},
width=1*\linewidth,
height=0.75*\columnwidth,
title style={font=\footnotesize,align=center},
legend cell align=left,
legend style={font=\footnotesize},
legend columns=3,
legend style={at={(0.45,-0.175)},anchor=north,draw=none,/tikz/every even column/.append style={column sep=-0.15mm},cells={align=left}},
ylabel near ticks,
xlabel near ticks,
xmin=1,
xmax=3.0,
ymin=1e-6,
ymax=1e-1,
xlabel={\textcolor{black}{$E_b/N_0$, dB}},
ylabel={\textcolor{black}{BLER}},
grid=both,
major grid style={solid,draw=gray!50},
minor grid style={densely dotted,draw=gray!50},
label style={font=\footnotesize},
tick label style={font=\footnotesize},
]

\addplot[myparula51] table[x=ebn0,y=P] {figures/512_256_hybrid_awgn_L=32.txt};
\addlegendentry{Code-$1$ $L=32$}
\addplot[myparula41] table[x=ebn0,y=P] {figures/512_256_hybrid_small_list_awgn_L=32.txt};
\addlegendentry{Code-2 $L=32$}
\addplot[myparula61] table[x=ebn0,y=P] {figures/512_256_proposed_awgn_L=32.txt};
\addlegendentry{Code-3 $L=32$}

\addplot[myparula54] table[x=ebn0,y=P] {figures/512_256_hybrid_awgn_L=128.txt};
\addlegendentry{Code-1 $L=128$}
\addplot[myparula44] table[x=ebn0,y=P] {figures/512_256_hybrid_small_list_awgn_L=128.txt};
\addlegendentry{Code-2 $L=128$}
\addplot[myparula64] table[x=ebn0,y=P] {figures/512_256_proposed_awgn_L=128.txt};
\addlegendentry{Code-3 $L=128$}

\addplot[myparula55] table[x=ebn0,y=P] {figures/512_256_hybrid_awgn_L=1024.txt};
\addlegendentry{Code-1 $L=1024$}
\addplot[myparula45] table[x=ebn0,y=P] {figures/512_256_hybrid_small_list_awgn_L=1024.txt};
\addlegendentry{Code-2 $L=1024$}

\addplot[color=myParula06LightBlue,line width=1pt,solid] table[x=ebn0,y=P] {figures/512_256_proposed_ML_LB_awgn.txt};
\addlegendentry{Code-3 ML LB}
\addplot[color=myParula05Green,line width=1pt,solid] table[x=ebn0,y=P] {figures/512_256_hybrid_ML_LB_awgn.txt};
\addlegendentry{Code-1 ML LB}
\addplot[color=myParula04Purple,line width=1pt,solid] table[x=ebn0,y=P] {figures/512_256_hybrid_small_list_ML_LB_awgn.txt};
\addlegendentry{Code-2 ML LB}

\addplot[color=black,line width=1pt,dashed] table[x=snr,y=P] {figures/512_256_awgn_rcu.txt};
\addlegendentry{RCU \textcolor{black}{$\boldsymbol{\cdot\cdot\cdot\cdot\cdot}$} MC}
\addplot[color=black,line width=1pt,dotted] table[x=snr,y=P] {figures/512_256_awgn_mc.txt};

\end{semilogyaxis}
\end{tikzpicture}

%% file: figures/proposed_vs_5G_512_256_CRC7.tex
\begin{tikzpicture}
\begin{semilogyaxis}[
mark options={solid,scale=0.75},
width=1*\linewidth,
height=0.75*\columnwidth,
title style={font=\footnotesize,align=center},
legend cell align=left,
legend style={font=\footnotesize},
legend columns=3,
legend style={at={(0.45,-0.175)},anchor=north,draw=none,/tikz/every even column/.append style={column sep=-0.15mm},cells={align=left}},
ylabel near ticks,
xlabel near ticks,
xmin=1,
xmax=3.0,
ymin=1e-6,
ymax=1e-1,
xlabel={\textcolor{black}{$E_b/N_0$, dB}},
ylabel={\textcolor{black}{BLER}},
grid=both,
major grid style={solid,draw=gray!50},
minor grid style={densely dotted,draw=gray!50},
label style={font=\footnotesize},
tick label style={font=\footnotesize},
]

\addplot[myparula32] table[x=ebn0,y=P] {figures/512_256_5G_awgn_L=8_CRC7.txt};
\addlegendentry{Polar + CRC-$7$ $L=8$}
\addplot[myparula62] table[x=ebn0,y=P] {figures/512_256_proposed_awgn_L=8.txt};
\addlegendentry{Code-3 $L=8$}
\addplot[color=black,line width=1pt,dashed] table[x=snr,y=P] {figures/512_256_awgn_rcu.txt};
\addlegendentry{RCU}

\addplot[myparula31] table[x=ebn0,y=P] {figures/512_256_5G_awgn_L=32_CRC7.txt};
\addlegendentry{Polar + CRC-$7$ $L=32$}
\addplot[myparula61] table[x=ebn0,y=P] {figures/512_256_proposed_awgn_L=32.txt};
\addlegendentry{Code-3 $L=32$}
\addplot[color=black,line width=1pt,dotted] table[x=snr,y=P] {figures/512_256_awgn_mc.txt};
\addlegendentry{MC}

\addplot[myparula34] table[x=ebn0,y=P] {figures/512_256_5G_awgn_L=128_CRC7.txt};
\addlegendentry{Polar + CRC-$7$ $L=128$}
\addplot[myparula64] table[x=ebn0,y=P] {figures/512_256_proposed_awgn_L=128.txt};
\addlegendentry{Code-3 $L=128$}

\end{semilogyaxis}
\end{tikzpicture}

%% file: figures/proposed_vs_5G_512_256.tex
\begin{tikzpicture}
\begin{semilogyaxis}[
mark options={solid,scale=0.75},
width=1*\linewidth,
height=0.75*\columnwidth,
title style={font=\footnotesize,align=center},
legend cell align=left,
legend style={font=\footnotesize},
legend columns=3,
legend style={at={(0.45,-0.175)},anchor=north,draw=none,/tikz/every even column/.append style={column sep=1mm},cells={align=left}},
ylabel near ticks,
xlabel near ticks,
xmin=1,
xmax=3.0,
ymin=5e-7,
ymax=1e-1,
xlabel={\textcolor{black}{$E_b/N_0$, dB}},
ylabel={\textcolor{black}{BLER}},
grid=both,
major grid style={solid,draw=gray!50},
minor grid style={densely dotted,draw=gray!50},
label style={font=\footnotesize},
tick label style={font=\footnotesize},
]

\addplot[myparula32] table[x=ebn0,y=P] {figures/512_256_5G_awgn_L=8.txt};
\addlegendentry{5G $L=8$}
\addplot[myparula31] table[x=ebn0,y=P] {figures/512_256_5G_awgn_L=32.txt};
\addlegendentry{5G $L=32$}
\addplot[myparula34] table[x=ebn0,y=P] {figures/512_256_5G_awgn_L=128_CRC11.txt};
\addlegendentry{5G $L=128$}
\addplot[myparula42] table[x=ebn0,y=P] {figures/512_256_hybrid_small_list_awgn_L=8.txt};
\addlegendentry{Code-2 $L=8$}

\addplot[myparula41] table[x=ebn0,y=P] {figures/512_256_hybrid_small_list_awgn_L=32.txt};
\addlegendentry{Code-2 $L=32$}

\addplot[myparula44] table[x=ebn0,y=P] {figures/512_256_hybrid_small_list_awgn_L=128.txt};
\addlegendentry{Code-2 $L=128$}

\addplot[myparula35] table[x=ebn0,y=P] {figures/512_256_5G_awgn_L=1024_CRC11.txt};
\addlegendentry{5G $L=1024$}
\addplot[myparula45] table[x=ebn0,y=P] {figures/512_256_hybrid_small_list_awgn_L=1024.txt};
\addlegendentry{Code-2 $L=1024$}


\addplot[color=black,line width=1pt,dashed] table[x=snr,y=P] {figures/512_256_awgn_rcu.txt};
\addlegendentry{RCU \textcolor{black}{$\boldsymbol{\cdot\cdot\cdot\cdot\cdot}$} MC}
\addplot[color=black,line width=1pt,dotted] table[x=snr,y=P] {figures/512_256_awgn_mc.txt};

\end{semilogyaxis}
\end{tikzpicture}

%% file: figures/proposed_vs_5G_512_256_CRC16.tex
\begin{tikzpicture}
\begin{semilogyaxis}[
mark options={solid,scale=0.75},
width=1*\linewidth,
height=0.75*\columnwidth,
title style={font=\footnotesize,align=center},
legend cell align=left,
legend style={font=\footnotesize},
legend columns=3,
legend style={at={(0.45,-0.175)},anchor=north,draw=none,/tikz/every even column/.append style={column sep=-0.15mm},cells={align=left}},
ylabel near ticks,
xlabel near ticks,
xmin=1,
xmax=3.0,
ymin=1e-6,
ymax=1e-1,
xlabel={\textcolor{black}{$E_b/N_0$, dB}},
ylabel={\textcolor{black}{BLER}},
grid=both,
major grid style={solid,draw=gray!50},
minor grid style={densely dotted,draw=gray!50},
label style={font=\footnotesize},
tick label style={font=\footnotesize},
]

\addplot[myparula31] table[x=ebn0,y=P] {figures/512_256_5G_awgn_L=32_CRC16.txt};
\addlegendentry{Polar + CRC-$16$ $L=32$}
\addplot[myparula51] table[x=ebn0,y=P] {figures/512_256_hybrid_awgn_L=32.txt};
\addlegendentry{Code-1 $L=32$}
\addplot[color=black,line width=1pt,dashed] table[x=snr,y=P] {figures/512_256_awgn_rcu.txt};
\addlegendentry{RCU}

\addplot[myparula34] table[x=ebn0,y=P] {figures/512_256_5G_awgn_L=128_CRC16.txt};
\addlegendentry{Polar + CRC-$16$ $L=128$}
\addplot[myparula54] table[x=ebn0,y=P] {figures/512_256_hybrid_awgn_L=128.txt};
\addlegendentry{Code-1 $L=128$}
\addplot[color=black,line width=1pt,dotted] table[x=snr,y=P] {figures/512_256_awgn_mc.txt};
\addlegendentry{MC}

\addplot[myparula35] table[x=ebn0,y=P] {figures/512_256_5G_awgn_L=1024_CRC16.txt};
\addlegendentry{Polar + CRC-$16$ $L=1024$}
\addplot[myparula55] table[x=ebn0,y=P] {figures/512_256_hybrid_awgn_L=1024.txt};
\addlegendentry{Code-1 $L=1024$}

\end{semilogyaxis}
\end{tikzpicture}

%% file: figures/inac_evo_512.tex
\begin{tikzpicture}
\footnotesize
\begin{axis}[
mark options={solid,scale=1},
width=1*\linewidth,
height=0.55*\columnwidth,
legend cell align=left,
legend columns=1,
ylabel near ticks,
xlabel near ticks,
xmin=1,
xmax=512,
ymin=0,
xlabel={\textcolor{black}{\own{Decoding stage} $m$}},
ylabel={\textcolor{black}{$\Bar{D}_m$}},
grid=both,
major grid style={solid,draw=gray!50},
minor grid style={densely dotted,draw=gray!50},
]

\addplot[color=myParula07Red,line width = 1.5pt,dotted] table[x=index,y=entropy] {figures/results_ExpectedEqRM_n=512_BEC.txt};
\addlegendentry{Sim.}

\addplot[color=myParula01Blue,line width = 1pt,dashed] table[x=index,y=entropy] {figures/results_qm_RM_n=512.txt};
\addlegendentry{\edd{Via} \eqref{eq:Markov_app}}
\addplot[color=myParula05Green,line width = 1pt] table[x=index,y=entropy] {figures/results_qsm_RM_n=512.txt};
\addlegendentry{\eqref{eq:recursive_app}}

\end{axis}
\end{tikzpicture}

%% file: figures/inac_evo_8192.tex
\begin{tikzpicture}
\footnotesize
\begin{axis}[
mark options={solid,scale=1},
width=1*\linewidth,
height=0.55*\columnwidth,
legend cell align=left,
legend columns=1,
ylabel near ticks,
xlabel near ticks,
xmin=1,
xmax=8192,
xtick={1600,3200,4800,6400,8000},
ymin=0,
xlabel={\textcolor{black}{\own{Decoding stage} $m$}},
ylabel={\textcolor{black}{$\Bar{D}_m$}},
grid=both,
major grid style={solid,draw=gray!50},
minor grid style={densely dotted,draw=gray!50},
]

\addplot[color=myParula07Red,line width = 1.5pt,dotted] table[x=index,y=entropy] {figures/results_ExpectedEqRM_n=8192_BEC=0.48.txt};
\addlegendentry{Sim.}

\addplot[color=myParula01Blue,line width = 1pt,dashed] table[x=index,y=entropy] {figures/results_qm_RM_n=8192.txt};
\addlegendentry{\edd{Via} \eqref{eq:Markov_app}}
\addplot[color=myParula05Green,line width = 1pt] table[x=index,y=entropy] {figures/results_qsm_RM_n=8192.txt};
\addlegendentry{\eqref{eq:recursive_app}}

\end{axis}
\end{tikzpicture}

%% file: figures/inac_evo.tex
\begin{tikzpicture}
\begin{axis}[
mark options={solid,scale=0.75},
width=1*\linewidth,
height=0.55*\columnwidth,
title style={font=\footnotesize,align=center},
legend cell align=left,
legend style={font=\footnotesize},
legend columns=4,
legend style={at={(0.5,-0.25)},anchor=north,draw=none,/tikz/every even column/.append style={column sep=3mm},cells={align=left}},
ylabel near ticks,
xlabel near ticks,
xmin=1,
xmax=512,
ymin=0,
ymax=35,
xlabel={\textcolor{black}{\own{Decoding stage} $m$}},
ylabel={\textcolor{black}{\scriptsize{$\ed{\bar{D}_m}$}}},
grid=both,
major grid style={solid,draw=gray!50},
minor grid style={densely dotted,draw=gray!50},
label style={font=\footnotesize},
tick label style={font=\footnotesize},
]

\addplot[color=myParula07Red,line width = 1.5pt, dotted] table[x=index,y=entropy] {figures/results_ExpectedEqRM_n=512.txt};
\addlegendentry{Sim.}

\addplot[color=myParula05Green,line width = 1pt,solid] table[x=index,y=entropy] {figures/results_qsm_RM_n=512.txt};
\addlegendentry{\eqref{eq:recursive_app}}
\addplot[color=myParula01Blue,line width = 1pt,dashed] table[x=index,y=entropy] {figures/results_ql_RM_n=512.txt};
\addlegendentry{$90\%$ Conf. \edd{via} \eqref{eq:Markov_app}}
\addplot[color=gray,line width = 0.5pt,dash dot] table[x=index,y=entropy] {figures/results_trace1_RM_n=512.txt};
\addlegendentry{Traces Sim.}
\addplot[color=myParula01Blue,line width = 1pt,dashed] table[x=index,y=entropy] {figures/results_qu_RM_n=512.txt};

\addplot[color=gray,line width = 0.5pt,dotted] table[x=index,y=entropy] {figures/results_trace2_RM_n=512.txt};
\addplot[color=gray,line width = 0.5pt,dotted] table[x=index,y=entropy] {figures/results_trace3_RM_n=512.txt};
\addplot[color=gray,line width = 0.5pt,dotted] table[x=index,y=entropy] {figures/results_trace4_RM_n=512.txt};
\addplot[color=gray,line width = 0.5pt,dotted] table[x=index,y=entropy] {figures/results_trace5_RM_n=512.txt};
\addplot[color=gray,line width = 0.5pt,dotted] table[x=index,y=entropy] {figures/results_trace6_RM_n=512.txt};
\addplot[color=gray,line width = 0.5pt,dotted] table[x=index,y=entropy] {figures/results_trace7_RM_n=512.txt};
\addplot[color=gray,line width = 0.5pt,dotted] table[x=index,y=entropy] {figures/results_trace8_RM_n=512.txt};
\addplot[color=gray,line width = 0.5pt,dotted] table[x=index,y=entropy] {figures/results_trace9_RM_n=512.txt};
\addplot[color=gray,line width = 0.5pt,dotted] table[x=index,y=entropy] {figures/results_trace10_RM_n=512.txt};
\addplot[color=gray,line width = 0.5pt,dotted] table[x=index,y=entropy] {figures/results_trace11_RM_n=512.txt};
\addplot[color=gray,line width = 0.5pt,dotted] table[x=index,y=entropy] {figures/results_trace12_RM_n=512.txt};
\addplot[color=gray,line width = 0.5pt,dotted] table[x=index,y=entropy] {figures/results_trace13_RM_n=512.txt};
\addplot[color=gray,line width = 0.5pt,dotted] table[x=index,y=entropy] {figures/results_trace14_RM_n=512.txt};
\addplot[color=gray,line width = 0.5pt,dotted] table[x=index,y=entropy] {figures/results_trace15_RM_n=512.txt};

\end{axis}
\end{tikzpicture}

%% file: figures/list_growth_rate.tex
\begin{tikzpicture}
\begin{axis}[
mark options={solid,scale=0.75},
width=1*\linewidth,
height=0.55*\columnwidth,
title style={font=\footnotesize,align=center},
legend cell align=left,
legend style={font=\footnotesize},
legend columns=4,
legend style={at={(0.5,-0.25)},anchor=north,draw=none,/tikz/every even column/.append style={column sep=1mm},cells={align=left}},
ylabel near ticks,
xlabel near ticks,
xmin=0,
xmax=1,
ymin=0,
xlabel={\textcolor{black}{\own{Normalized decoding stage} $w$}},
ylabel={\textcolor{black}{\scriptsize{$\frac{1}{N}\ed{\bar{D}}_{wN}$}}},
yticklabel style={/pgf/number format/.cd,fixed,precision=2},
xtick = {0.2,0.4,0.6,0.8,1},
grid=both,
major grid style={solid,draw=gray!50},
minor grid style={densely dotted,draw=gray!50},
label style={font=\footnotesize},
tick label style={font=\footnotesize},
]

\addplot[color=myParula01Blue,line width = 0.5pt,solid] table[x=index,y=entropy] {figures/results_qsm_RM_n=512_SNR=0.48.txt};
\addlegendentry{$N=2^{9}$}
\addplot[color=myParula02Orange,line width = 0.5pt,solid] table[x=index,y=entropy] {figures/results_qsm_RM_n=2048_SNR=0.48.txt};
\addlegendentry{$N=2^{11}$}
\addplot[color=myParula03Yellow,line width = 0.5pt,solid] table[x=index,y=entropy] {figures/results_qsm_RM_n=8192_SNR=0.48.txt};
\addlegendentry{$N=2^{13}$}
\addplot[color=myParula04Purple,line width = 0.5pt,solid] table[x=index,y=entropy] {figures/results_qsm_RM_n=32768_SNR=0.48.txt};
\addlegendentry{$N=2^{15}$}
\addplot[color=myParula05Green,line width = 0.5pt,solid] table[x=index,y=entropy] {figures/results_qsm_RM_n=131072_SNR=0.48.txt};
\addlegendentry{$N=2^{17}$}
\addplot[color=myParula06LightBlue,line width = 0.5pt,solid] table[x=index,y=entropy] {figures/results_qsm_RM_n=524288_SNR=0.48.txt};
\addlegendentry{$N=2^{19}$}
\addplot[color=myParula07Red,line width = 0.5pt,solid] table[x=index,y=entropy] {figures/results_qsm_RM_n=2097152_SNR=0.48.txt};
\addlegendentry{$N=2^{21}$}
\addplot[color=black,line width = 0.5pt,solid] table[x=index,y=entropy] {figures/results_qsm_RM_n=8388608_SNR=0.48.txt};
\addlegendentry{$N=2^{23}$}
\addplot[color=olive,line width = 0.5pt,solid] table[x=index,y=entropy] {figures/results_qsm_RM_n=33554432_SNR=0.48.txt};
\addlegendentry{$N=2^{25}$}
\addplot[color=magenta,line width = 0.5pt,solid] table[x=index,y=entropy] {figures/results_qsm_RM_n=134217728_SNR=0.48.txt};
\addlegendentry{$N=2^{27}$}
\addplot[color=cyan,line width = 0.5pt,solid] table[x=index,y=entropy] {figures/results_qsm_RM_n=536870912_SNR=0.48.txt};
\addlegendentry{$N=2^{29}$}
\addplot[color=pink,line width = 0.5pt,solid] table[x=index,y=entropy] {figures/results_qsm_RM_n=2147483648_SNR=0.48.txt};
\addlegendentry{$N=2^{31}$}
\addplot[color=myParula01Blue,line width = 1.5pt,dotted] table[x=index,y=entropy] {figures/sim_results_qsm_RM_n=512_SNR=0.48.txt};
\addplot[color=myParula02Orange,line width = 1.5pt,dotted] table[x=index,y=entropy] {figures/sim_results_qsm_RM_n=2048_SNR=0.48.txt};
\addplot[color=myParula03Yellow,line width = 1.5pt,dotted] table[x=index,y=entropy] {figures/sim_results_qsm_RM_n=8192_SNR=0.48.txt};

\end{axis}
\end{tikzpicture}

%% file: figures/dm_distribution.tex
\begin{tikzpicture}
\begin{axis}[
mark options={solid,scale=0.75},
width=1*\linewidth,
height=0.55*\columnwidth,
title style={font=\footnotesize,align=center},
legend cell align=left,
legend style={font=\footnotesize},
legend columns=1,
ylabel near ticks,
xlabel near ticks,
xmin=0,
xmax=0.1,
ymin=0,
ymax=0.08,
xlabel={\textcolor{black}{Normalized subspace dimension $d$}},
ylabel={\textcolor{black}{$\Pr\left(\frac{1}{N}D_{\left[0.4N\right]} = d\right)$}},
xticklabel style={
           /pgf/number format/fixed,
            /pgf/number format/fixed zerofill,
            /pgf/number format/precision=2
        },
xtick = {0.02,0.04,0.06,0.08,0.10},
grid=both,
major grid style={solid,draw=gray!50},
minor grid style={densely dotted,draw=gray!50},
label style={font=\footnotesize},
tick label style={font=\footnotesize},
]

\addplot[only marks,mark=x,color=myParula01Blue] table[x=index,y=prob] {figures/subsdimAtMI_n=512Erasure=0.48.txt};
\addlegendentry{$N=2^{9}$}
\addplot[only marks,mark=x,color=myParula02Orange] table[x=index,y=prob] {figures/subsdimAtMI_n=2048Erasure=0.48.txt};
\addlegendentry{$N=2^{11}$}
\addplot[only marks,mark=x,color=myParula03Yellow] table[x=index,y=prob] {figures/subsdimAtMI_n=8192Erasure=0.48.txt};
\addlegendentry{$N=2^{13}$}

\end{axis}
\end{tikzpicture}

%% file: figures/RMsubcode_bec.tex
\begin{tikzpicture}
\begin{semilogyaxis}[
mark options={solid,scale=0.75},
width=1*\linewidth,
height=0.75*\columnwidth,
title style={font=\footnotesize,align=center},
legend cell align=left,
legend style={font=\footnotesize},
legend pos=north east,
ylabel near ticks,
xlabel near ticks,
xmin=0.36,
xmax=0.48,
x dir=reverse,
ymin=1e-6,
ymax=1e-1,
xlabel={\textcolor{black}{Erasure probability $\epsilon$}},
ylabel={\textcolor{black}{BLER}},
grid=both,
major grid style={solid,draw=gray!50},
minor grid style={densely dotted,draw=gray!50},
label style={font=\footnotesize},
tick label style={font=\footnotesize},
]

\addplot[myparula51] table[x=erasure,y=P] {figures/512_256_hybridsubcode_bec.txt};
\addlegendentry{Code-$1$}
\addplot[myparula22] table[x=erasure,y=P] {figures/512_256_RM_bec.txt};
\addlegendentry{RM$(4,9)$}
\addplot[myparula74] table[x=erasure,y=P] {figures/512_256_RMsubcode_bec.txt};
\addlegendentry{\own{\ac{dRM}}$(4,9)$}
\addplot[color=black,line width = 1pt,dashed] table[x=erasure,y=P] {figures/512_256_berlekamp_bec.txt};
\addlegendentry{\own{BRC}}
\addplot[color=black,line width = 1pt,dotted] table[x=erasure,y=P] {figures/512_256_singleton_bec.txt};
\addlegendentry{\own{Singleton}}

\end{semilogyaxis}
\end{tikzpicture}

%% file: figures/inac_evo_RM_vs_design.tex
\begin{tikzpicture}
\begin{axis}[
mark options={solid,scale=0.75},
width=1*\linewidth,
height=0.55*\columnwidth,
title style={font=\footnotesize,align=center},
legend cell align=left,
legend style={font=\footnotesize},
ylabel near ticks,
xlabel near ticks,
xmin=1,
xmax=512,
ymin=0,
xlabel={\textcolor{black}{\own{Decoding stage} $m$}},
ylabel={\textcolor{black}{\scriptsize{$\ed{\bar{D}_m}$}}},
grid=both,
major grid style={solid,draw=gray!50},
minor grid style={densely dotted,draw=gray!50},
label style={font=\footnotesize},
tick label style={font=\footnotesize},
]

\addplot[color=myParula07Red,line width = 1pt,solid] table[x=index,y=entropy] {figures/results_qsm_RM_n=512.txt};
\addlegendentry{\own{\ac{dRM}}}
\addplot[color=myParula05Green,line width = 1pt,solid] table[x=index,y=entropy] {figures/results_qs_Code1_n=512.txt};
\addlegendentry{Code-$1$}

\addplot[color=myParula07Red,line width = 1.5pt,dotted] table[x=index,y=entropy] {figures/results_ExpectedEqRM_n=512.txt};
\addplot[color=myParula05Green,line width = 1.5pt,dotted] table[x=index,y=entropy] {figures/results_ExpectedEqCode1_n=512.txt};

\end{axis}
\end{tikzpicture}

%% file: tex/Conclusions.tex
\section{Conclusion}\label{sec:conc}
In this paper, we consider the theoretical question ``What list size suffices to achieve \own{maximum-likelihood} decoding performance under \own{successive cancellation list (SCL)} \ed{decoding}?''. The results identify information-theoretic quantities associated with the required list size \ed{on average} and also lead to \own{bounds} that can be computed efficiently even for very long codes. \own{It has been shown that t}he logarithm of the random required list size concentrates around the mean. A simple \own{and accurate} recursive approximation of this mean is provided for the \own{binary erasure channel (BEC)}.

Simulation results show that this approximation captures the dynamics of the required list size at each stage of decoding on the BEC. For general \own{binary memoryless symmetric} channels, e.g., \own{binary-input additive white Gaussian noise (BIAWGN) channel}, the analysis identified the key quantity $\ed{\Bar{D}_m}$ as a proxy for the uncertainty in SCL decoding. The analysis suggested codes with improved performance under SCL decoding with various list sizes, e.g., $L\in\edd{[8,1024]}$, for short- to moderate-length codes, over the BIAWGN channel\edd{, although the work primarily aims at understanding \ac{SCL} decoding from an information-theoretic perspective}. Future work should devise a constructive algorithm which take into account the introduced quantities to design good codes for SCL decoding with practical list sizes for \own{a wide} range of block lengths, e.g., for $N\leq 2048$.

%% file: tex/appendix.tex
\appendix
\subsection{The Designs for $(512,256)$ Case}
Let $\mathcal{A}_{\mathrm{PW}}$ denote the set of information positions of the $(512,256)$ polar code designed according to \ac{PW} with $\beta=2^{\nicefrac{1}{4}}$\cite{beta_polar}. The sets $\mathcal{A}_1$, $\mathcal{A}_2$ and $\mathcal{A}_3$ corresponding to Code-$1$, Code-$2$ and Code-$3$, respectively, are given as follows:
\begin{align}
    \mathcal{A}_3 &= (\mathcal{A}_{\mathrm{PW}}\setminus\{449,450,451,453,457,465,481\})\\
    &\qquad\qquad\qquad\qquad\cup\{1,64,118,122,159,200,284\}, \\
    \mathcal{A}_2 &= (\mathcal{A}_3\setminus\{122,421,425,433\})\cup\{32,174,272,280\}, \\
    \mathcal{A}_1 &= (\mathcal{A}_2\setminus\{64,96,125,180,418,419\})\\
    &\qquad\qquad\qquad\qquad\cup\{48,56,94,108,122,152\}.
\end{align}
For all codes, each frozen bit is set to a random linear combination of preceding information bit(s). The performance curves will also be available on \cite{pretty_good_codes}.